%% file: flspp.tex
\documentclass[a4paper,UKenglish,cleveref, autoref, thm-restate]{lipics-v2021}

\usepackage[T1]{fontenc}
\usepackage{lmodern}
\def\doi#1{\href{https://doi.org/\detokenize{#1}}{\url{https://doi.org/\detokenize{#1}}}}

\usepackage{bbm}
\usepackage{multicol}
\usepackage{float}
\usepackage{graphicx}
\usepackage{listings}
\usepackage{amsmath,amssymb}
\usepackage{mathtools}
\usepackage{todonotes}
\usepackage{tikz}
\usepackage{caption}
\usepackage{subcaption}
\usepackage{wrapfig}
\usepackage{svg}
\usepackage[algoruled,boxed,lined,linesnumbered]{algorithm2e}
\usepackage{multirow}
\usepackage{adjustbox} 
\usepackage{tabularx} 
\usepackage{hyperref}
\usepackage{booktabs}
\usepackage{tensor} 
\usepackage{cellspace} 
\nolinenumbers

\DeclareMathOperator {\bigO}{\mathcal{O}}

\SetKwInOut{Parameter}{Parameters}

\newcommand{\cl}{\phi}
\newcommand{\llo}{C,\cl}

\author{Theo Conrads}{Department of Computer Science, University of Cologne, Germany}{}{}{}
\author{Lukas Drexler}{Faculty of Mathematics and Natural Sciences, Department of Computer Science,\\ Heinrich Heine University Düsseldorf, Germany}{lukas.drexler@hhu.de}{https://orcid.org/0000-0001-9395-6711}{Gefördert durch die Deutsche Forschungsgemeinschaft (DFG) -- Projektnummer 456558332 --- funded by the Deutsche Forschungsgemeinschaft (DFG, German Research Foundation) -- Projektnummer 456558332.}
\author{Joshua K\"onen}{Institute of Computer Science, University of Bonn, Germany}{s6jjkoen@uni-bonn.de}{https://orcid.org/0000-0003-4245-4812}{funded by Lamarr Institute for Machine Learning and Artificial Intelligence --- lamarr-institute.org}
\author{Daniel R. Schmidt}{Faculty of Mathematics and Natural Sciences, Department of Computer Science\\ Heinrich Heine University Düsseldorf, Germany}{dschmidt@hhu.de}{https://orcid.org/0000-0001-7381-912X}{}
\author{Melanie Schmidt}{Faculty of Mathematics and Natural Sciences, Department of Computer Science,\\ Heinrich Heine University Düsseldorf, Germany}{mschmidt@hhu.de}{https://orcid.org/0000-0003-4856-3905}{Gefördert durch die Deutsche Forschungsgemeinschaft (DFG) -- Projektnummer 456558332 --- funded by the Deutsche Forschungsgemeinschaft (DFG, German Research Foundation) -- Projektnummer 456558332.}

\begin{document}
\title{Local Search k-means++ with Foresight}
\authorrunning{T. Conrads et al.}

\Copyright{Theo Conrads, \and Lukas Drexler, \and Joshua K\"onen, \and Daniel R. Schmidt, \and Melanie Schmidt} 


\supplement{\url{https://github.com/algo-hhu/FLSpp}}
\EventEditors{John Q. Open and Joan R. Access}
\EventNoEds{2}
\EventLongTitle{42nd Conference on Very Important Topics (CVIT 2016)}
\EventShortTitle{CVIT 2016}
\EventAcronym{CVIT}
\EventYear{2016}
\EventDate{December 24--27, 2016}
\EventLocation{Little Whinging, United Kingdom}
\EventLogo{}
\SeriesVolume{42}
\ArticleNo{23}

\maketitle
\begin{abstract}
Since its introduction in 1957, Lloyd's algorithm for $k$-means clustering has been extensively studied and has undergone several improvements. While in its original form it does not guarantee any approximation factor at all, Arthur and Vassilvitskii (SODA 2007) proposed $k$-means++ which enhances Lloyd's algorithm by a seeding method which guarantees a $\bigO (\log k)$-approximation in expectation. More recently, Lattanzi and Sohler (ICML 2019) proposed LS++ which further improves the solution quality of $k$-means++ by local search techniques to obtain a $\bigO(1)$-approximation.
On the practical side, the $\emph{greedy}$ variant of $k$-means++ is often used although its worst-case behaviour is provably worse than for the standard $k$-means++ variant.

We investigate how to improve LS++ further in practice. We study two options for improving the practical performance: (a) Combining LS++ with  greedy $k$-means++ instead of $k$-means++, and (b) Improving LS++ by better entangling it with Lloyd's algorithm. 
Option (a) worsens the theoretical guarantees of $k$-means++ but improves the practical quality also in combination with LS++ as we confirm in our experiments. Option (b) is our new algorithm, Foresight LS++. 
We experimentally show that FLS++ improves upon the solution quality of LS++. It retains its asymptotic runtime and its worst-case approximation bounds.

\end{abstract}
\clearpage
\setcounter{page}{1}
\input{introduction.tex}
\input{algorithm.tex}
\input{results.tex}
\clearpage

\bibliographystyle{splncs04}
\bibliography{bibliography}{}

\newpage
\input{appendix.tex}

\end{document}

%% file: introduction.tex
\section{Introduction}

In the vast area of clustering, one of the most popular approaches is min-sum-of-squared-error clustering, as modeled by the \emph{$k$-means} cost function. Given a set of vectors $X = \{x_1, \ldots, x_n\}\subset \mathbb{R}^d$ and a number of clusters $k \in \mathbb{N}$, the aim of $k$-means clustering is to find a set $C=\{c_1,\ldots,c_k\}$ of $k$ centers that minimizes
$
 \sum_{i=1}^n \min_{j=1,\ldots,k} || x_i - c_j||^2,
$
i.e., the sum of the squared distances of all points to their respective closest center in $C$.
For half a century, the most known algorithm for minimizing the $k$-means cost function was a local search heuristic often called \emph{$k$-means algorithm} or \emph{Lloyd's algorithm}, developed in 1957. The main steps of the algorithm are as follows:
\begin{enumerate}
 \item Find an initial set of $k$ centers, e.g., randomly chosen from $X$.
 \item Repeat a given number $s$ times or until convergence: \label{lloyds-algo}
 \begin{enumerate}
  \item For every $x \in X$, find a closest center $\cl(x)$ (ties broken arbitrarily) and use this information to form clusters $C_1,\ldots,C_k$ by $C_i = \cl^{-1}(c_i)$.

  \item 
  For all $C_i$, compute the \emph{mean}
  $
   \mu(C_i) = \frac{1}{|C_i|}\sum_{x \in C_i} x
  $
  and then replace $C$ by $C_{\text{new}} = \{\mu(C_1),\ldots,\mu(C_k)\}$.
 \end{enumerate}
\end{enumerate}
Notice that this description ignores edge cases like clusters which \lq run empty\rq\ during step 2(a). We also do not discuss variations of other stopping criteria here. 
It is well documented \cite{CelebiKV13,franti2019,pena1999} that the implementation of the first step of Lloyd's algorithm is crucial because bad initial centers can lead to bad local optima. It is also easy to find theoretical worst case examples where the local optimum is arbitrarily bad~\cite{KanungoMNPSW04}.

In 2007, Arthur and Vassilvitskii~\cite{AV07} proposed a method to significantly improve the first step of Lloyd's algorithm, leading to a new de-facto-standard algorithm for the $k$-means problem. The core idea is to choose the initial centers by an adaptive sampling procedure known as 
$d^2$-sampling initialization (also see Algorithm~\ref{d2sampling} below):
\begin{enumerate}
 \item Choose $c_1$ uniformly at random from $X$.
 \item For $i=2,\ldots,k$:
 \begin{itemize}
  \item For all $x \in X$, compute $p(x) := \frac{\min_{c \in C} ||x-c||^2}{\sum_{y\in X} \min_{c \in C} ||y-c||^2}$
  \item Sample a point $c_i$ where every $x \in X$ has probability $p(x)$
 \end{itemize}
  \item Return $C=\{c_1,\ldots,c_k\}$
\end{enumerate}
The algorithm proposed in~\cite{AV07} first computes $C$ by this routine and then runs Lloyd's algorithm with $C$ as the initial center set. This combination of $d^2$-sampling and Lloyd's algorithm is called \emph{$k$-means++}.

Let us consider some background to better understand the advantages of $k$-means++ (and the subsequent improvements): There are two main reasons why a clustering obtained with Lloyd's algorithm may be bad. Firstly, it may be that the underlying structure of the data does not fit the $k$-means objective, e.g., because the points are not well-clusterable with spherical clusters, or because we chose the wrong $k$. In this case, we should choose a different clustering method like kernel $k$-means, hierarchical clustering methods or density based clustering. But secondly, it may be that although optimal clusters with respect to the $k$-means objective are indeed perfect for us, Lloyd's algorithm does not find them because it converges to a local optimum.
Figure~\ref{chamechaude} gives a visual example for this: We see a $k$-means based image compression with four colors (The input pixels are represented by three-dimensional points based on the RGB-values, these points are clustered, and then the color of every pixel is replaced by the mean color of its cluster. Clusters are not contiguous within the picture). In this example, we can see how the difference in $k$-means cost is indeed reflected by the different quality of the compressed image.
\begin{figure}
    \centering
    \begin{subfigure}[b]{0.3\textwidth}
        \centering
        \includegraphics[width=\linewidth]{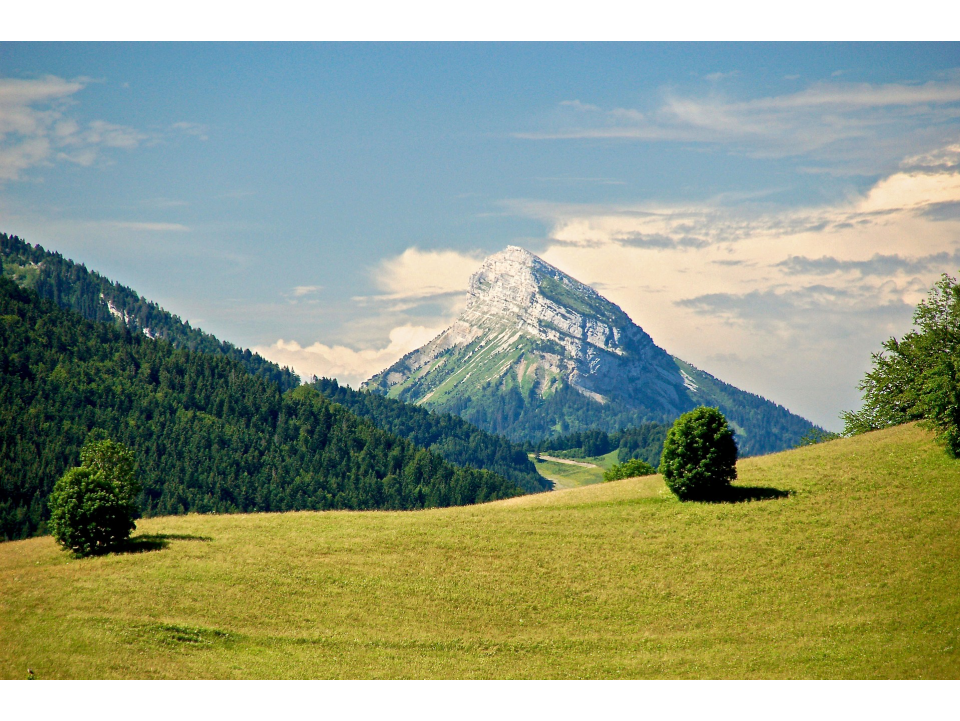}
        \caption{Original image}
    \end{subfigure}
    \hfill
    \begin{subfigure}[b]{0.3\textwidth}
        \centering
        \includegraphics[width=\linewidth]{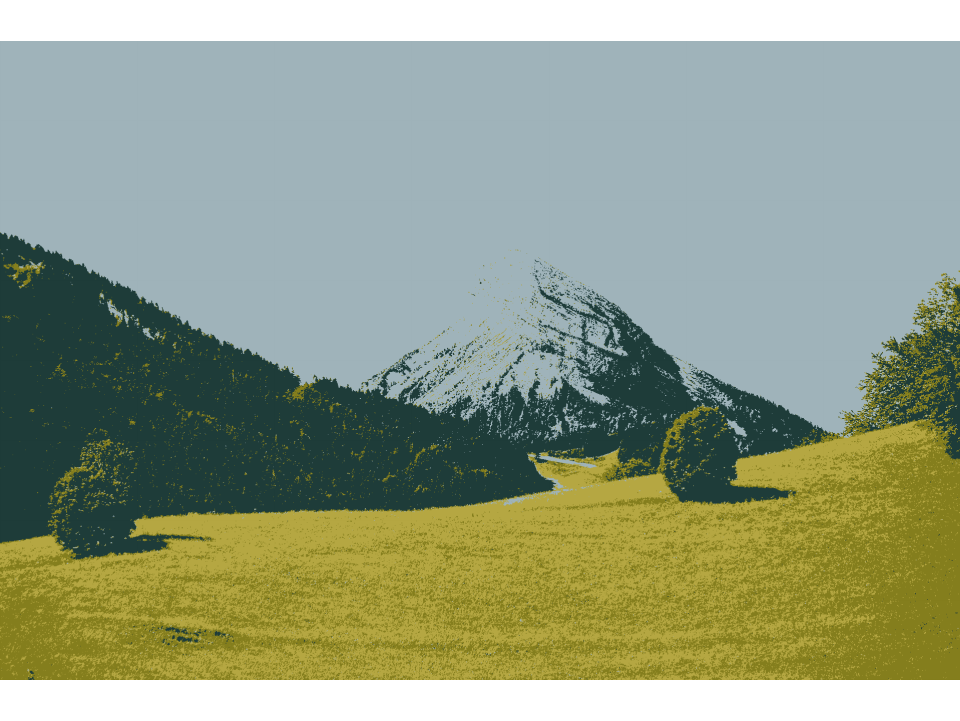}
        \caption{A local optimum}
    \end{subfigure}
    \hfill
    \begin{subfigure}[b]{0.3\textwidth}
        \centering
        \includegraphics[width=\linewidth]{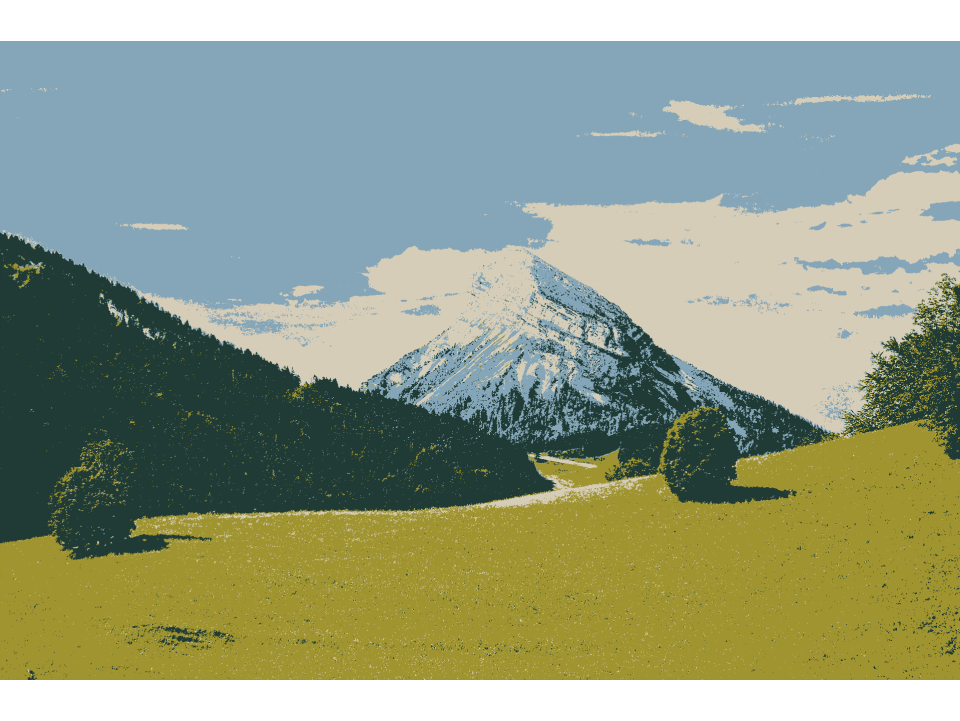}
        \caption{A better solution}
    \end{subfigure}
    \caption{Compression of an image with $k=4$ centers (i.e. colors). Subfigure (a) shows the original image. Subfigure (b) shows a local optimum with a $k$-means cost of $55.18\cdot 10^8$. We found this local optimum in runs of Lloyd with uniform initialization and in single runs of $k$-means++.
    Subfigure (c) shows a solution with a $k$-means cost of $43.09 \cdot 10^8$ (for example found by FLS++).\vspace*{-0.4cm}}
\label{chamechaude}
\end{figure}

Provided that the cost can actually guide us to a good solution, we know that there are hidden optimal clusters which we desire to find, and that we should improve upon the $k$-means cost to find them.
The idea of $k$-means++ is to find one point from each optimal cluster. This point should be \lq relatively\rq{} good, in the sense that it should allow Lloyd's algorithm to find a good center for the optimum cluster later. We call an optimum cluster \emph{covered} when the $d^2$-sampling sampled one of its points. Now, observe that clusters that are \emph{not} covered induce a high cost. By basing the $d^2$-sampling probabilities on the cost of points, the algorithm strives to find good points from uncovered optimum clusters.
Indeed, $k$-means++ succeeds with this goal most of the time, but it can miss some of the optimum clusters -- which then remain uncovered because the subsequent Lloyd steps are (in general) unable to shift a mean point from one optimum cluster to another.
However, $k$-means++ still computes a $\Theta(\log k)$-approximation on expectation, and it improves the practical performance of Lloyd's algorithm vastly.
\begin{figure}
\centering
\begin{subfigure}[t]{0.49\textwidth}
\includegraphics[scale=0.4]{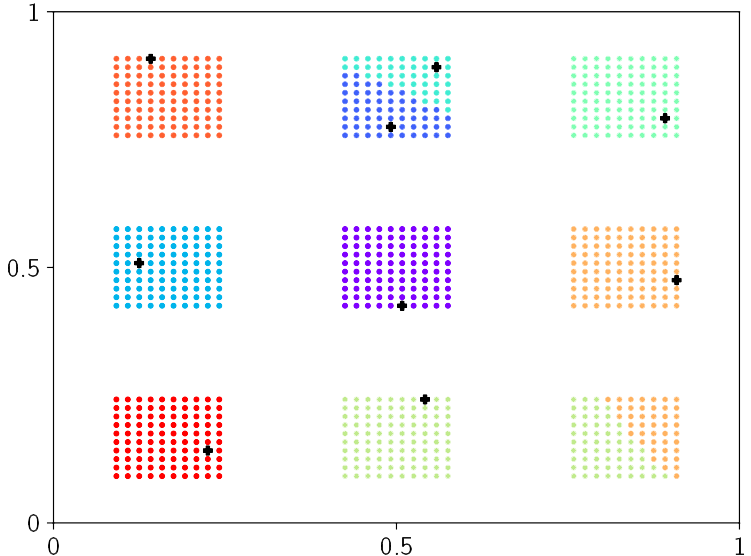}
\caption{A typical $d^2$-sampling initialization\label{d-two-sampling-solution-picture}}
\end{subfigure}\hfill
\begin{subfigure}[t]{0.49\textwidth}
\includegraphics[scale=0.4]{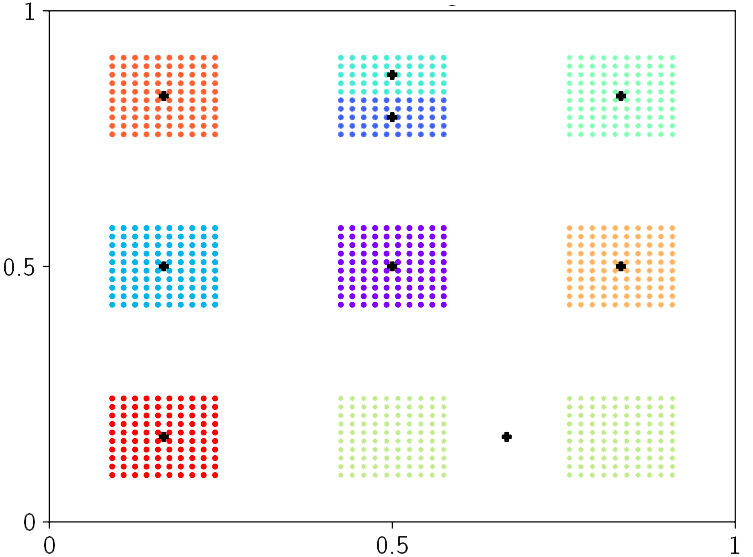}
\caption{Converging Lloyd's algorithm on (a)}
\end{subfigure}
\caption{This data set is by Fritzke \cite{fritzke2017boxes}, the illustration by Conrads \cite{C21}. The left side shows nine centers sampled by one run of $k$-means++ and the corresponding induced clusters are illustrated by colors. The right side shows how the clusters and centers look after running Lloyd's algorithm to convergence with the nine centers from the left as input.\label{fig-fritzke}\vspace*{-0.4cm}}
\end{figure}
Figure~\ref{fig-fritzke} illustrates how $k$-means++ can end up in a local optimum. We see that out of the nine natural clusters in the data set, $k$-means++ has covered eight -- the upper middle cluster was hit by two samples, and thus there was no center left in the end to cover the lower right cluster. Running Lloyd's algorithm on this center set then results in nicely placed centers which capture the natural clusters perfectly in seven of nine cases. However, Lloyd's algorithm cannot repair the mistake of not finding the lower right cluster.

It is now around fifteen years after the development of $k$-means++, and there has been increasing interest in the question of improving upon its theoretical and practical properties. There are two main ways to improve the behaviour explained in Figure~\ref{fig-fritzke}. Option A is to try to avoid missing clusters. It was already proposed in~\cite{AV07} to improve $k$-means++ by adding a greedy procedure: When choosing $c_i$, compute $\ell$ \emph{candidates}, all chosen according to the same probability distribution, and then out of these, pick the point which decreases the current cost the most. This variant of $k$-means++ is called \emph{greedy $k$-means++}. Although it has been observed as performing better in practice in some instances (as in~\cite{AV07} and in~\cite{CelebiKV13}), greedy $k$-means++ actually has a \emph{worse} worst-case behavior than $k$-means++. Bhattacharya et al. \cite{BERM19} give a family of point sets where the expected quality of solutions computed by greedy $k$-means++ is lower bounded by $\Omega(\ell \log k)$. This is because points which look beneficial to the greedy procedure may actually be very bad centers, for instance, if points lie \lq on the boundary\rq\ between clusters. If two clusters are uncovered and there is a point in the middle, then choosing this point decreases the overall cost more than any point closer to the mean of the two clusters. However, a point in the middle of two clusters and far away from their means is certainly a bad choice for a cluster center.
Recently, \cite{grunau2023} showed a better lower bound of $\Omega(\ell^3 \log^3 k /\log^2(\ell \log k))$ while also giving an upper bound of $\bigO(\ell^3 \log^3 k)$.
Still, greedy $k$-means++ is for example the default initialization method in Python's scikit learn package~\cite{scikit-learn} and
our experiments confirm that this is a justified choice.

\begin{figure}
	\begin{subfigure}[t]{0.475\textwidth}
	\centering
	\includegraphics[width=0.6\linewidth,page=1]{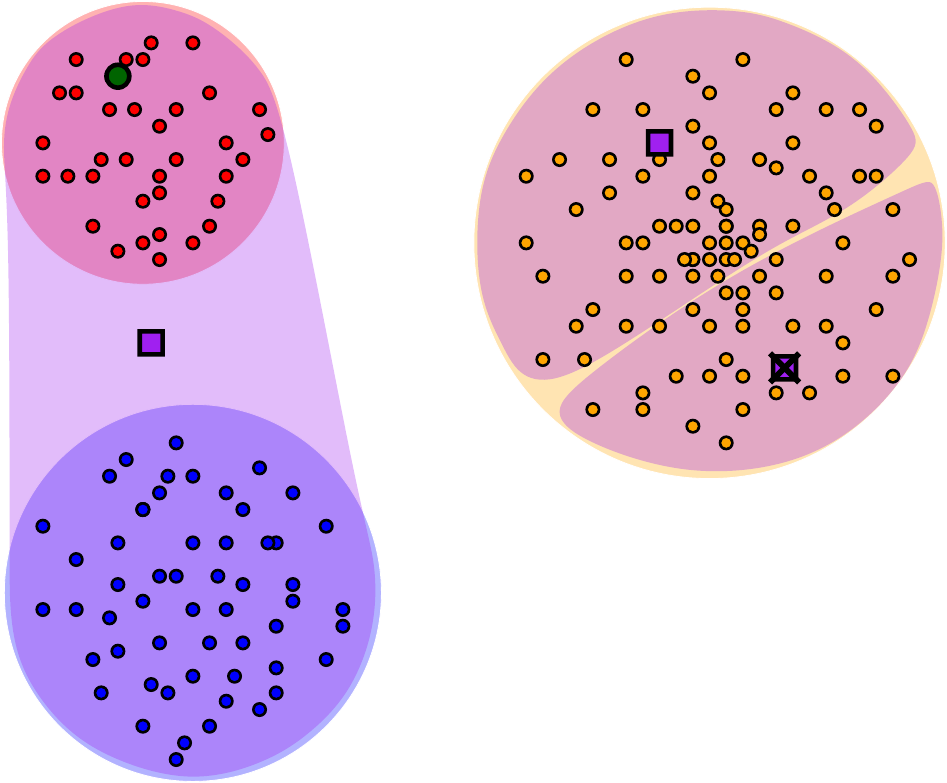}
	\caption{Initial converged clustering (without the circle-shaped green point).
	The circle-shaped green point is sampled, the crossed out center is removed.
	}
	\end{subfigure} \hfill
	\begin{subfigure}[t]{0.475\textwidth}
        \centering
        \includegraphics[width=0.6\linewidth,page=2]{figures/heuristic_clustercost_motivation2.pdf}
        \caption{After performing the center swap, Lloyd's iterations will converge to the optimal clustering.}
    \end{subfigure}%
    \caption{Improving solutions with local search steps.\label{fig:lsimprovement}\vspace*{-0.4cm}}
\end{figure}
Another option is to perform improvement steps on the $k$-means++ solution before running Lloyd's algorithm. A prominent technique in the theoretical analysis of the $k$-means problem (and other clustering problems) is to do local search by swapping in and out centers. A solution set $C$ is improved by taking out a constant number $t$ of centers and replacing them by $t$ points from $X$. \cite{KanungoMNPSW04} show that this approach with sufficiently large $t$ leads to a $(9+\epsilon)$-approximation for the $k$-means problem. However, the approach is fairly slow since even for $t=1$, checking if there is an improving swap takes a lot of time. Lattanzi and Sohler \cite{LS19} combine local search for $t=1$ with adaptive sampling. Instead of checking all $x \in X$ for improving swaps, they (multiple times) sample a point from $X$ by adaptive sampling and check if exchanging it for a center in $C$ improves the solution. They show that starting from an initial $k$-means++ solution, this yields a $\bigO(1)$-approximation after $\bigO(k \log\log k)$ steps. The resulting algorithm is called LS++. The analysis has since been improved by Choo et al. in~\cite{choo2020}, where it is shown that, for any $\varepsilon > 0$, performing $\varepsilon k$ steps yields a $\bigO (1/\varepsilon^3)$-approximation. Figure~\ref{fig:lsimprovement} shows an idealized visualization of how improvement steps shall work: In a solution where one optimum cluster got covered twice and other clusters have been clustered together, one wants to find a point from an uncovered optimal cluster with $d^2$-sampling and then converge to an optimum solution with Lloyd's algorithm.

Our aim in this paper is to improve the practical quality of both greedy $k$-means++ and LS++. We start with LS++, but our idea is to \emph{combine} it with the strength of Lloyd's algorithm instead of simply running local improvement steps ahead of it.
We propose \emph{Foresight-LS++ (FLS++)} combining steps of Lloyd's algorithm and local search:
When exchanging centers, it is not at all clear which centers should be swapped because we do not know the effect of the swap when Lloyd's algorithm runs later. Does a particular swap really enable Lloyd's algorithm to find good centers? The best way to answer this question is to actually perform one step of Lloyd's algorithm for a potential swap. This way, instead of comparing the cost of the solution before and after the swap directly, we add some \emph{foresight}: We perform one step of Lloyd's algorithm and see how the solution develops depending on the swaps. 
The candidates for new centers are also found with $d^2$-sampling, so the centers we sample follow the same distribution. 
\begin{wrapfigure}{r}{6.5cm}
 \centering
 \includegraphics[width=0.45\textwidth]{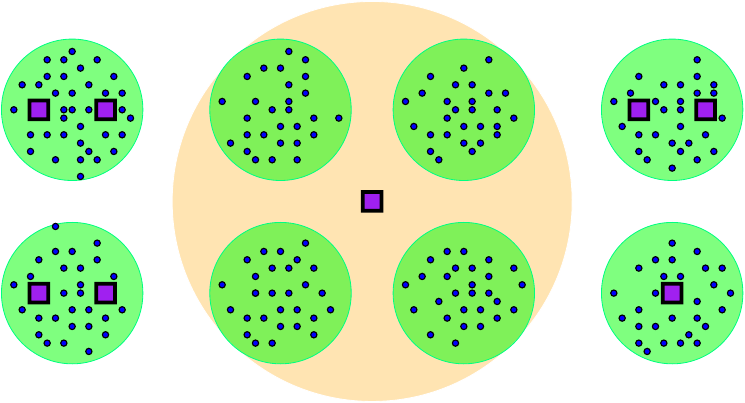}
 \caption{A example with eight optimal clusters (green). When swapping in a new center, the best center to delete is the one in the middle: without it, Lloyd's algorithm can repair the solution.
 \label{fig:other-case}}
\end{wrapfigure}
With our procedure, 
we may avoid putting in unfavorable centers. However, the biggest impact of our change is finding the best center to \emph{remove} in a more informed way. 
More precisely, some of the swaps take out superfluous centers in clusters that were hit multiple times. However, there are also beneficial swaps of a different type: Here, we swap out a center that is actually the only one covering a cluster, just to better distribute the centers in that area. Figure~\ref{fig:other-case} shows a schematic illustration of situations where this can be indeed helpful. By running one step of Lloyd's algorithm, FLS++ can identify both (and potentially more) types of beneficial swaps.

\subparagraph*{Further related work}
The $k$-means problem is NP-hard even for $k=2$~\cite{dasgupta2008hardness} and $d=2$~\cite{mahajan2009planar}, 
but there exist different constant-factor approximation algorithms.
The currently best approximation ratio is $6.357$, achieved by Ahmadian et al.~\cite{AhmadianNSW20} with an LP rounding based approach. Prior to their work, the best known approximation ratio of $9+\epsilon$ was achieved via a local search approach analyzed by Kanungo et al.~\cite{KanungoMNPSW04}. Awasthi et al.~\cite{AwasthiCKS15} show that for general $k$ and $d$ there exists some $c > 0$ such that it is NP-hard to approximate $k$-means with factor $c$. The constant $c$ is at least $1.0013$ as shown by Lee et al.~\cite{LeeSW17}. 
The problem gets easier when \emph{bicriterial} solutions are allowed, i.e., when we may select more than $k$ centers and still compare to the cost of a solution with $k$ centers. There are multiple bicriterial approximations, and Wei~\cite{NIPS2016_357a6fdf} indeed shows that for any constant $\beta > 1$, selecting $\beta k$ centers with $D^2$-sampling gives a constant factor approximation (where the constant depends on $\beta$).

There is a vast literature on Lloyd's algorithm and various speed-up techniques. A good survey on initialization methods is due to Celebi et al.~\cite{CelebiKV13} who give an extensive evaluation of different heuristics to initialize Lloyd's algorithm, including $k$-means++ and a greedy variant of it. LS++ was developed after that survey, so it is not included in the comparison. Examples for speed-up techniques are the works by Elkan~\cite{Elkan03} and Hamerly~\cite{Hamerly12}. These speed up the execution of Lloyd's algorithm but do not change the outcome. 

%% file: algorithm.tex
\subparagraph*{Notation} 
The distance of a point $p$ to a set of centers $C$ is $d(p,C) = \min_{c\in C} ||p-c||$, and for a set of points $P$ we define the distance to a set of centers as $d(P,C) = \sum_{p\in P} d(p,C) = \sum_{p\in P} \min_{c\in C}||p-c||$. We refer to the \emph{cost} of $P$ with respect to center set $C$ as $\Phi(P,C) = d^2(P,C) = \sum_{p \in P} d^2(p,C) = \sum_{p\in P} \min_{c\in C}||p-c||^2$. 
For an arbitrary assignment $\cl':P\rightarrow C$ we define $\Phi(P,C,\cl')=\sum_{p\in P}d^2(p,\cl'(p))$ as the cost of a clustering $C$ using assignment $\cl'$.

\section{Algorithms}

We first describe the $k$-means++ algorithm. It is stated in Algorithm~\ref{kmeans++} and explained below. 
\begin{minipage}[t]{7.3cm}
  \vspace{0pt}  
\begin{algorithm}[H]
 \KwIn{Point set $P\subset \mathbb{R}^d$, number $k\in \mathbb{N}$}
 \KwOut{Center set $C\subset \mathbb{R}^d$}
 \SetKwFunction{samplecenter}{SampleCenter}
 \BlankLine
 \everypar={\nl}
 $C = \emptyset$\\
    \For {$i=1$ \textbf{to} $k$}{
        $c = \samplecenter (P,C)$\\
        $C = C\cup \{c\}$
    }
    \Return $C$
    \caption{$d^2$-\texttt{sampling-init}}
    \label{d2sampling}
\end{algorithm}
\end{minipage}%
\hspace*{0.1cm}
\begin{minipage}[t]{6cm}
  \vspace{0pt}
\begin{algorithm}[H]
 \KwIn{Point set $P\subset \mathbb{R}^d$, numbers $k,s\in \mathbb{N}$}
 \KwOut{Center set $C\subset \mathbb{R}^d$}
 \SetKwFunction{samplecenter}{$d^2$-\texttt{sampling-init}}
 \SetKwFunction{lloyd}{Lloyd}
 \BlankLine
 \everypar={\nl}
 $C = \samplecenter(P,k)$\\
 $\llo = \lloyd(P,C,s)$\\
    \Return $C$
    \caption{$k$-means++}
    \label{kmeans++}
\end{algorithm}
\end{minipage}\\
The $k$-means++ algorithm consists of sampling $k$ initial centers with $d^2$-sampling and then heuristically refining the solution by Lloyd's method. For ease of notation in Algorithm~\ref{algo_fls++}, the method \texttt{LLoyd}$(P,C,s)$ also outputs the optimal assignment of the points to the centers before the final centroids are computed, see the description of Lloyd's algorithm at the beginning of the introduction.
For the sampling part, we assume that we have access to a function $\texttt{SampleCenter}(P,C)$ which, given a point set $P$ and center set $C$, returns a point $p$ from $P$ that is sampled with probability $d^2(p,C)/d^2(P,C)$ (unless $C = \emptyset$, in which case the probability is $1/|P|$). This is the original $k$-means++ sampling, also referred to as $d^2$-sampling. Notice that a sample can be obtained in time $O(n)$ if all distances $d^2(p,C)$ are already known: First, add $d^2(p,C)$ for all $p\in P$ to obtain the sum $S$. Second, draw a uniform sample $r$ from $[0,1]$. Third, iterate through $P$ and add up the distances again until the sum exceeds $r \cdot S$ for the first time; take the point before the point for which this happened. It is important that one needs all values of $d^2(p,C)$ to do so; one can compute them in time $\bigO(ndk)$ when needed but it is beneficial to store values.
Repeating $\texttt{SampleCenter}(P,C)$ for $k$ iterations yields the initialization part of $k$-means++, see Algorithm~\ref{d2sampling}. 
We also assume that we have access to a subroutine for Lloyd's method. Given a set $P$ of input points, a set $C$ of initial centers and a maximum number of steps $s$, the method $\texttt{Lloyd}(P,C,s)$ runs the loop in Step 2 of Lloyd's algorithm as described on page~\pageref{lloyds-algo} to compute and return a set $C_{\text{out}}$ of centers. We define $\Phi(\texttt{Lloyd}(P,C,s)) = \Phi(P,C_{\text{out}})$. 
Given these subroutines, $k$-means++ can be described by the pseudo code in Algorithm~\ref{kmeans++}.

\subparagraph*{LS++ and FLS++}
Algorithm~\ref{FLS++} shows the pseudo code of FLS++. Without Lines 2, 5 and 7 and with 
directly comparing $\Phi(P,(C\setminus\{c_{\text{old}}\})\cup \{c_{\text{new}}\})$ to $\Phi(P,C^{\text{min}})$ in Line 8, the code describes LS++. The core idea of LS++ is to further improve the initial solution by local search steps before running Lloyd's algorithm.  LS++ first samples $k$ centers with $\texttt{d}^2$-$\texttt{sampling-init}(P,k)$. Then it does the following $Z$ times: Pick another point $c_{\text{new}}$ with $d^2$-sampling, then find the best way to swap out a center $c_{\text{old}}$ and replace it with $c_{\text{new}}$. The parameter $Z$ is assumed to be at least $100000 k \log \log k$ in the theoretical analysis in \cite{LS19}, but set to values $\le 25$ in the practical evaluation~\cite{LS19}. After $Z$ improvement steps, LS++ calls Lloyd's algorithm.
LS++ does indeed improve the solution quality. But the strength in the original $k$-means++ algorithm lies in the fact that it uses the power of Lloyd's to quickly do local refinements: $d^2$-sampling itself would not be competitive, it is the combination of finding clusters and refining the centers that makes $k$-means++ such a success in practical applications. Consequently, we believe that increasing the symbiosis of sampling and Lloyd's algorithm is the key to improving LS++.

The additional lines of FLS++ call Lloyd's algorithm for one iteration, which means that all centers are replaced by the centroids of their clusters. Line 2 uses this to refine the initial sampling solution, Line 5 makes sure that we always keep centroids (instead of mixing in input points into the center set), and Line 7 is the step that gave the algorithm the prefix \emph{foresight}: Instead of checking the cost of $(C \backslash\{c_{\text{old}}\}) \cup \{c_{\text{new}}\}$, we run a Lloyd's step to evaluate how the solution cost will develop after updating the centers.
\begin{wrapfigure}{l}{7.2cm}
\centering
\begin{minipage}[t]{7cm}\centering
\begin{algorithm}[H]
  \KwIn{Point set $P\subset \mathbb{R}^d$, numbers $k,s\in \mathbb{N}$}
  \Parameter{$Z\in \mathbb{N}_{0}$}
  \KwOut{Center set $C\subset \mathbb{R}^d$}
  \SetKwFunction{sampling}{$d^2$-sampling}
  \SetKwFunction{samplecenter}{SampleCenter}
  \SetKwFunction{lloyd}{Lloyd}
  \BlankLine
  \everypar={\nl}
  $C = \sampling (P,k)$\\
  $\llo = \texttt{Lloyd}(P,C,1)$\\
  \For{$i=1$ \KwTo $Z$}{
  $c_{\text{new}} = \samplecenter (P,C)$\\
  $C^{\min}, \cl^{\text{min}} = \texttt{Lloyd}(P,C,1)$\\
    \For{$c_{\text{old}} \in C$}{
        $C', \cl' = \lloyd (P,(C\setminus\{c_{\text{old}}\})\cup \{c_{\text{new}}\},1)$\\

        \If{$\Phi(P,C',\cl') < \Phi(P,C^{\min}, \cl^{\textnormal{min}}) $}{
            $C^{\min}, \cl^{\min} = C',\cl'$      
        }
    }
    $C = C^{\min}$
  }
  $C, \cl=\lloyd (P,C,s)$\\
  \Return $C$
\caption{\texttt{Foresight-LS++}}\label{FLS++}
\label{algo_fls++}
\end{algorithm}
\end{minipage}\vspace*{-0.5cm}
\end{wrapfigure}
This makes it much easier for FLS++ to identify swaps that remove a superfluous center like the extra center in the upper cluster in Figure~\ref{d-two-sampling-solution-picture}.  

One can view FLS++ either as a modified LS++ algorithm or as a modified Lloyd's algorithm. Viewed as a modified Lloyd's algorithm one can think of FLS++ as performing Lloyd's algorithm but doing a center swap in every iteration to avoid running into local optima. Viewed as a LS++ modification, one can think of it as running Lloyd's algorithm for $1$ step between the local search improvement steps.
Why do we run exactly one step of Lloyd's algorithm? The reason is that this can be done very efficiently. Running multiple Lloyd's steps for \emph{each} $c_{\text{old}}$ would incur a running time of $O(ndk^2)$ for \emph{each} local search step. But recomputing the centroids can be done cleverly for all points together such that one local seach step in Algorithm~\ref{FLS++} takes $\bigO(ndk)$ just as for LS++. 

\begin{lemma}\label{lem:fasterruntime}
 One iteration of the For-Loop in Lines $3$-$13$ of Algorithm~\ref{FLS++} can be implemented such to run in time $\bigO(ndk)$.
\end{lemma}
\begin{proof}
First we compute the distances of every point $p$ to its closest center $\cl_1(p)$ and second-closest center $\cl_2(p)$ in time $\bigO(ndk)$.
With this information, we can sample a point $c_{\text{new}}$ with $d^2$-sampling in time $\bigO(n)$.
Next, we compute the distances of all points to the sample point $c_{\text{new}}$. Then we check all $k$ current centers whether we want to swap it out. For one candidate $c_{\text{old}}$, we run through all $n$ points $p$, and determine which is its closest center. We can do this in time $\bigO(1)$ in the following way: if $\cl_1(p)\neq c_{\text{new}}$, then we compare $d(p,\cl_1(p))$ with $d(p,c_{\text{new}})$.
If $\cl_1(p) = c_{\text{new}}$, then we compare $d(p,\cl_2(p))$ with $d(p,c_{\text{new}})$.
After doing this for all $c_{\text{old}}$ and all $p$, we have determined the cluster membership of all points in time $\bigO(kn)$. Finally, we want to recompute the centroids and costs of all clusters which we can do in time $\bigO(nd)$.
It is important to note that we do not recompute the reassignment of all points to the new set of centers as this would increase the overall running time $O(ndk^2)$. This is delayed until the next iteration and thus only done for the chosen $C^{\min}$. 
\end{proof}

The approximation guarantee is also not affected.

\begin{corollary}
    FLS++ with $Z\ge 100000 k \log \log k$ computes an $O(1)$-approximation. 
\end{corollary}
\begin{proof}
 \cite{LS19} shows that for any $C$ with cost larger than $500 \cdot OPT$ and a $c_{\text{new}}$ sampled by $d^2$-sampling, there is a $c_{\text{old}}$ such that the cost of $(C\backslash\{c_{\text{old}}\})\cup \{c_{\text{new}}\}$ compared to the cost of $C$ is smaller than that of $C$ by a fraction of $(1/100k)$ with probability $1/1000$ (Lemma 3 in \cite{LS19}). The analysis of the algorithm follows from this fact. Lloyd's steps are always improving, since the centroid is always the best center for a cluster. Thus Lemma 3 in \cite{LS19} still holds and thus does the approximation guarantee.
\end{proof}

%% file: results.tex
\section{Experimental Results}\label{results}

\subparagraph*{Setup.} 
The computations were performed on an Intel(R) Xeon(R) E3-1240 running at 3.7Ghz and 8 cores.
The code is written in C++.
The code and datasets used can be found at \url{https://github.com/lukasdrexler/flspp_code}. 
In the following we abbreviate $k$-means++ with KM++. We refer to the greedy variant of KM++ as GKM++ (analogously for (G)LS++ and (G)FLS++).

\subparagraph*{Datasets.} Our experiments are based on datasets used in~\cite{LS19}, image RGB-data used in~\cite{FS06}, some datasets where the optimal solution is known for specific values of $k$, and some new synthetically generated datasets similar to \emph{rectangles}.

\scalebox{0.725}{
\centering
\begin{tabular}{|c|c|c|c|}
\hline
dataset & number of points & dimension $d$ & source \\
\hline \hline
\textit{rectangles} & $1296$ & $2$ & \cite{FS18}\\
\hline
\textit{circles} & $10000$ & $2$ & -\\
\hline
\textit{close circles} & $10000$ & $2$ & -\\
\hline
\textit{pr91} & $2392$ & $2$ &  \cite{PR91}\\
\hline
\textit{D31} & $3100$ & $2$ & \cite{FS18}\\
\hline
\textit{s3} & $5000$ & $2$ & \cite{FS18}\\
\hline
\textit{A3} & $7500$ & $2$ & \cite{FS18} \\
\hline
\end{tabular}}\scalebox{0.725}{
\begin{tabular}{|c|c|c|c|}
\hline
dataset & number of points & dimension $d$ & source \\
\hline \hline
\textit{Tower} & $4915200$ & $3$ & \cite{FS06} \\
\hline 
\textit{Clegg} & $716320$ & $3$ & \cite{FS06} \\
\hline
\textit{Frymire} & $1235390$ & $3$ & \cite{FS06} \\
\hline
\textit{Body measurements} & $507$ & $5$ (reduced) & \cite{HPJK03} \\
\hline
\textit{Concrete strength} & $1030$ & $9$ & \cite{YEH19981797} \\
\hline
\textit{KDD Bio Train} & $145751$ & $74$ & \cite{KDD-2004} \\
\hline
\textit{KDD Phy Test} & $100000$ & $78$ & \cite{KDD-2004} \\
\hline
\end{tabular}}\\
For some datasets, the choice of $k$ is obvious, as there exists a clear ground truth clustering: \emph{rectangles} from Figure~\ref{fig-fritzke} consist of $k=36$ clusters;
likewise, the two generated datasets \textit{circles} and \textit{close circles} consist of $k=100$ separated clusters as visualized in Figure~\ref{fig:circle_circle} in the Appendix.
For other datasets, we evaluated several values of $k$: e.g. on \emph{pr91} where the optimum $k$-means cost is known for several values of $k$ (see~\cite{aloise2012}).

\subparagraph*{Overall performance on big data sets.}
To get an idea of the algorithms' strengths and weaknesses, we compare the results of all four algorithms on three large data sets, \emph{KDD Bio Train}, \emph{KDD Phy Test}, and \emph{Tower}. 
In Figure~\ref{fig:boxplot}, we compare KM++, GKM++, GLS++ with 25 local search steps (which is the number chosen in~\cite{LS19}) and GFLS++ with 5, 10, and 15 steps. Since GFLS++ performs one Lloyd iteration in every local search step, 25 such steps take longer than 25 steps in GLS++. Hence, we would like to compare how GFLS++ performs when using \emph{fewer} local search steps. We take a more detailed look at the trade-off between runtime and cost further below.
We perform $50$ \emph{runs} of each configuration, i.e., each algorithm is called $50$ times with a specific $k$ on each dataset. The initial centers for GLS++ and GFLS++ are chosen by greedy $d^2$-sampling. We execute  GKM++, GLS++ and GFLS++ with the same initial set of centers.

As is to be expected, a larger value of $Z$ increases runtime but decreases cost. On both \emph{Phy Test} and \emph{Tower}, GFLS++ with $Z=15$ obtains the overall minimum cost at the expense of a longer runtime, while outperforming GLS++ slightly at $Z=10$ in terms of cost \emph{and} runtime.
Table~\ref{cost_comp_large} shows the maximum, average and minimum cost that the algorithms produced on each data set with $k=100$. While the overall differences are not too big, GFLS++ consistently obtains the lowest values.
\begin{figure}[t]
  \centering
        \includegraphics[width=0.9\textwidth]{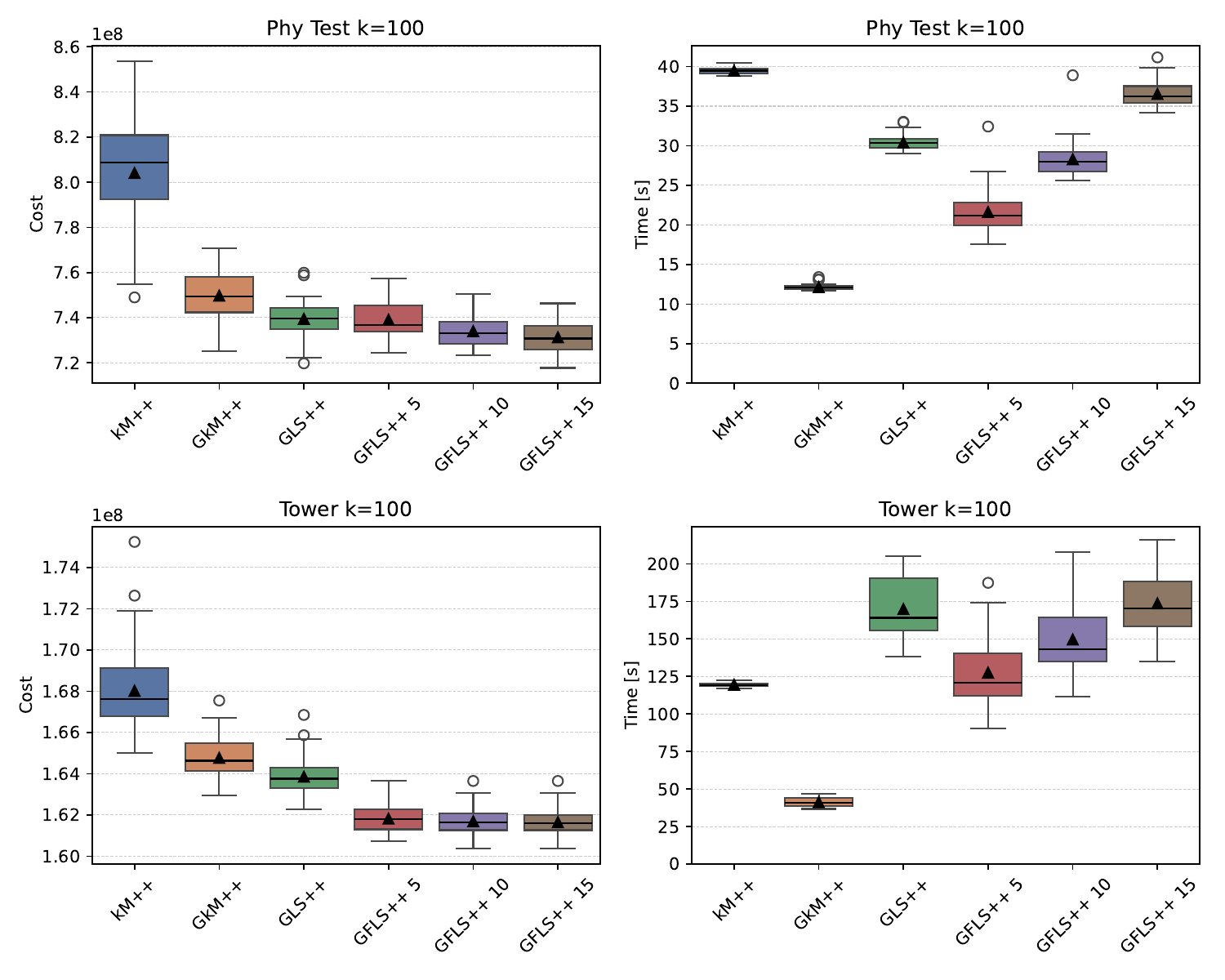}
    \caption{Comparison on two large datasets, for $R=50$ repetitions. GLS++ always performs 25 local search steps. For GFLS++, we display the results for performing 5, 10 and 15 such steps.}
\label{fig:boxplot}
\end{figure}
\begin{table}[b]
\centering 
\scalebox{0.85}{
    \begin{tabular}{lccccccc}
        \toprule
        \multirow{3}{*}{Datasets} & \multirow{3}{*}{} & \multicolumn{6}{c}{Algorithms} \\
        \cmidrule(lr){3-8}
        & & KM++ & GKM++ & LS++ & FLS++ 5 & FLS++ 10 & FLS++ 15 \\
        \midrule
        \multirow{3}{*}{Bio Train} & \textbf{Min} & 1.6292E+11 & 1.6258E+11 & 1.6244E+11 & 1.6116E+11 & 1.6116E+11 & \textbf{1.6110E+11} \\
                                   & \textbf{Mean} & 1.6487E+11 & 1.6311E+11 & 1.6302E+11 & 1.6162E+11 & 1.6162E+11 & \textbf{1.6160E+11} \\
                                   & \textbf{Max} & 1.6937E+11 & 1.6379E+11 & 1.6357E+11 & 1.6240E+11 & \textbf{1.6219E+11} & \textbf{1.6219E+11} \\
        \midrule
        \multirow{3}{*}{Phy Test} & \textbf{Min} & 7.4898E+08 & 7.2525E+08 & 7.1975E+08 & 7.2427E+08 & 7.2348E+08 & \textbf{7.1773E+08} \\
                                   & \textbf{Mean}& 8.0406E+08 & 7.4972E+08 & 7.3942E+08 & 7.3905E+08 & 7.3400E+08 & \textbf{7.3128E+08} \\
                                   & \textbf{Max} & 8.5367E+08 & 7.7058E+08 & 7.5988E+08 & 7.5721E+08 & 7.5032E+08 & \textbf{7.4630E+08} \\
        \midrule
        \multirow{3}{*}{Tower} & \textbf{Min} & 1.6500E+08 & 1.6296E+08 & 1.6228E+08 & 1.6073E+08 & \textbf{1.6036E+08} & \textbf{1.6036E+08}\\
                                   & \textbf{Mean} & 1.6802E+08 & 1.6477E+08 & 1.6385E+08 & 1.6182E+08 & 1.6169E+08 & \textbf{1.6164E+08} \\
                                   & \textbf{Max} & 1.7524E+08 & 1.6755E+08 & 1.6685E+08 & \textbf{1.6365E+08} & \textbf{1.6365E+08} & \textbf{1.6365E+08} \\
        \bottomrule
    \end{tabular}
}
\caption{Average cost comparison on large data sets with $k=100$ and $50$ runs.}
\label{cost_comp_large}
\end{table}
Runtimes and performance comparisons for different values of $k$ and for the \emph{KDD Bio Train} dataset can be seen in Section~\ref{sec:remaining_figures} in the appendix.

\subparagraph*{Best of repeated runs within a time limit.}\label{sec:fancy_compare}
\begin{table}[t]
\centering
\scalebox{.9625}{
\begin{adjustbox}{width=\textwidth}
\begin{tabular}{|c||Sc|Sc|Sc|c|c|}
\hline
data set & $c_{\text{KM++}}$ & $c_{\text{LS++}}$ & $c_{\text{FLS++}}$ & $C(\text{FLS++}, \text{KM++})$ & $C(\text{FLS++}, \text{LS++})$\\
\hline \hline
\textit{ pr91 } & $ 9.5637\text{{E+}}08 $ & $ 9.5276\text{{E+}}08 $ & $ 9.4696\text{{E+}}08 $ & $ 0.98\% $ & $ 0.61\% $ \\      \hline
\textit{ bio train features } & $ 2.3918\text{{E+}}11 $ & $ 2.3885\text{{E+}}11 $ & $ 2.3864\text{{E+}}11 $ & $ 0.23\% $ & $ 0.09\% $ \\      \hline
\textit{ concrete } & $ 3.2724\text{{E+}}06 $ & $ 3.1775\text{{E+}}06 $ & $ 3.1308\text{{E+}}06 $ & $ 4.33\% $ & $ 1.47\% $ \\      \hline
\textit{ circles } & $ 3.1913\text{{E+}}05 $ & $ 2.7164\text{{E+}}05 $ & $ 2.4773\text{{E+}}05 $ & $ 22.37\% $ & $ 8.8\% $ \\      \hline
\end{tabular}
\end{adjustbox}
}
\caption{Average cost. Last two columns show $(1-c_{\text{FLS++}}/c_B) \cdot 100\%$ for $B\in \{\text{KM++}, \text{LS++}\}$. }
\label{tab:normal_iterations_costs}
\end{table}

As we have seen in the previous section, solutions computed by GFLS++ have a smaller cost than GKM++ or GLS++ on average, but need more time to terminate.
In the following we want to analyze if running KM++ or LS++ multiple times can yield better results than GFLS++ in the same time.
Thus, we repeat each algorithm until a time limit is reached and for each algorithm only the best solution is returned.

We choose some iteration value $B\in \mathbb{N}$ and report the best solution found by exactly $B$ repetitions of GFLS++.
We then repeat (G)KM++ and GLS++ as long as their respective elapsed time is less or equal to the time used by FLS++; again, we return the best solution found.

The result of this comparison is shown in Table~\ref{tab:normal_iterations_costs}.
It shows the numerical values of the final average cost of each algorithm. Additionally we evaluated the relative cost difference when compared to GFLS++.
For two algorithms $A_1,A_2\in \{\text{KM++}, \text{LS++}, \text{FLS++}\}$ and their respective average costs $c_{A_1}, c_{A_2}\in \mathbb{R}_{> 0}$ we define their respective percentage cost difference as ${C(A_1,A_2):=(1-c_{A_1}/c_{A_2})\cdot 100\%}$.
On data set \textit{circles} the cost difference 
was the highest with approximately $22.37\%$ and $8.8\%$ when we compare the final average cost of FLS++ to KM++ and LS++.
This large difference also shows that even if clusters are well separated, LS++ might still fail to find an improvement if the optimal center centroids are not close to the actual data points.
In contrast, FLS++ can more often find an improvement because it evaluates the actual optimal centers for a new choice of centers, which brings the new centers closer to the actual optimal centroids.
For the other datasets the relative difference in cost is not large, but a positive amount which, if we take into account the number of times each algorithms returned the smallest cost, can be achieved with large success probability.

\subparagraph*{Performance over time.} 
We compare the best solutions found by all three algorithms and their average cost progression over the time.
We test this procedure $R\in \mathbb{N}$ times, where the $r$-th run being defined as:
\begin{itemize}
\item Run GFLS++ $B$ times, let $t_B^r$ be the used time until termination in run $r$.
\item Repeat GKM++ and GLS++ each as long as their respective elapsed time is at most $t_B^r$.
\end{itemize}
For run $r\in [R]$ let $\tensor*[^r]{c}{_A^t}$ for $A\in \{\text{KM++}, \text{LS++}, \text{FLS++}\}$ be the current minimum cost of algorithm $A$ at time $t\in \mathbb{R}_{\geq 0}$ and run $r$.
Let $t_{\min}$ be the first point in time where $A$ did terminate in every run $r\in[R]$, i.e., $^rc_A^t$ is always defined for every time $t\geq t_{\min}$.
The average cost of any algorithm $A$ after $R$ runs in time step $t\geq t_{\min}$ is defined as $\text{AVG}(A,t):=\frac{1}{R}\sum_{r=1}^R \tensor*[^r]{c}{_A^t}$.
As before, each algorithm starts in the $r$-th run with the same set of centers from the initialization process.

Figure~\ref{fig:algComp} shows the progress of each algorithm over time.
For LS++, we use $Z=25$ as number of repeat steps, as in the original experiments in \cite{LS19}.
We use the same value for FLS++ in the following experiments if not specified otherwise.
We repeat Lloyd steps, i.e., computing memberships of all points to their closest center and recomputing the new optimal centers for this clustering, until two consecutive
iterations produce solutions $C_1, C_2$ with $1- \frac{\Phi(P,C_2)}{\Phi(P,C_1)} < 0.0001$.
All plots are averaged over $100$ iterations, i.e., $R=100$, except for (b) which is averaged over $20$ iterations.
In all plots, FLS++ achieves the smallest final average cost and often beats the other algorithms
over the entire time frame. 
Similarly, LS++ is consistently better than KM++.
In (c) we also compare with the known optimal solution OPT: Despite the FLS++ improvements, there is still a significant gap.

\begin{figure}[t]
    \begin{subfigure}[b]{0.45\textwidth}
        \centering
        \includegraphics[width=\textwidth]{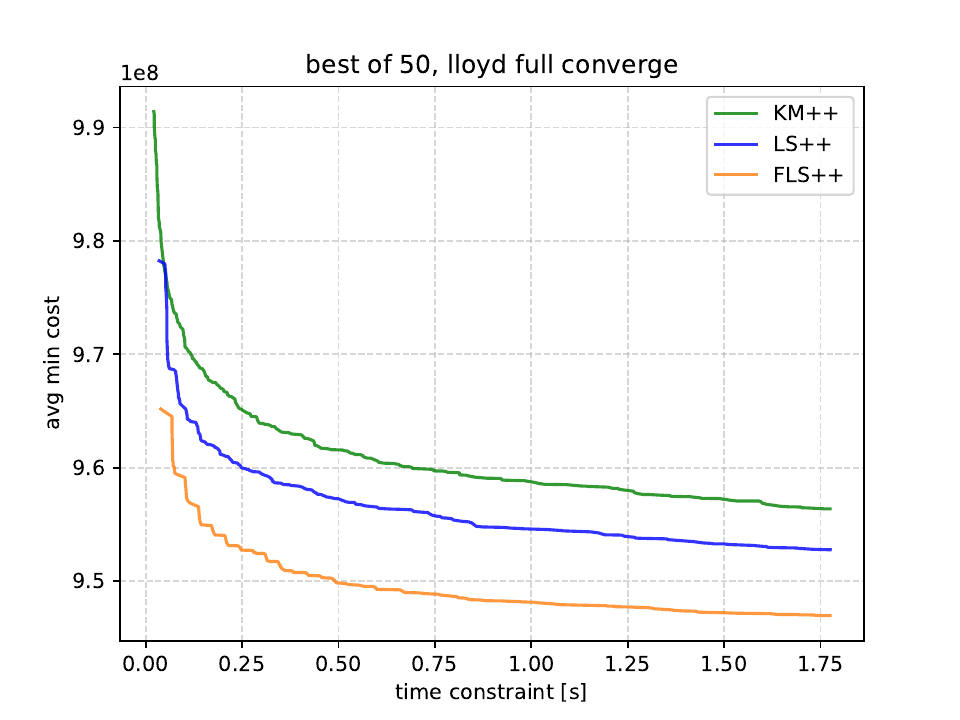}
        \caption{\textit{pr91}, $k=50$}
    \end{subfigure}%
    \begin{subfigure}[b]{0.45\textwidth}
        \centering
        \includegraphics[width=\textwidth]{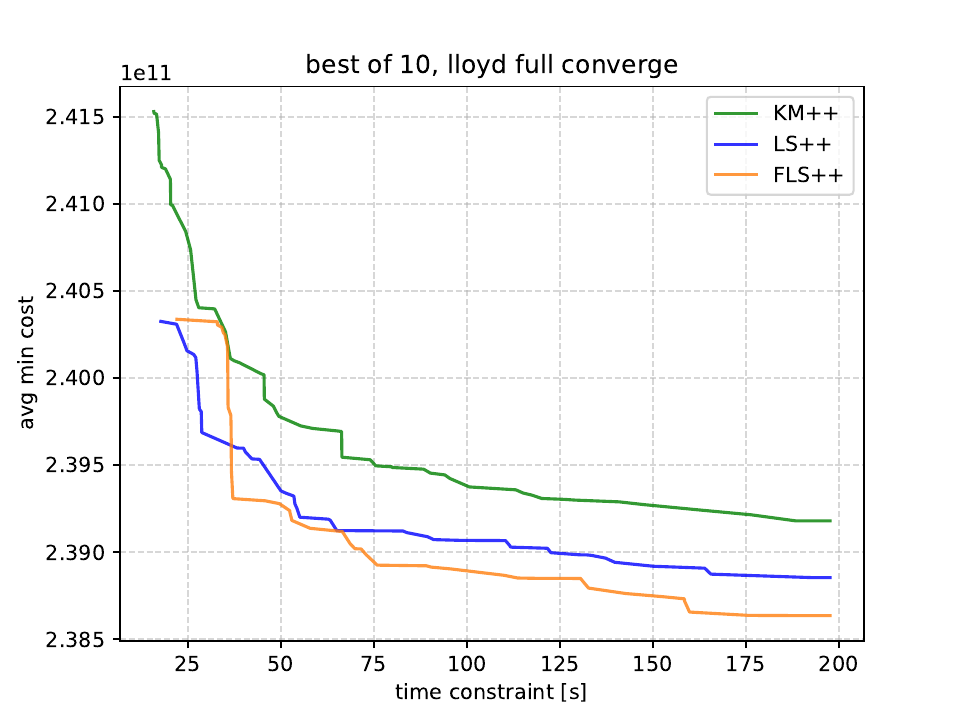}
        \caption{\textit{bio train features}, $k=25$}
    \end{subfigure}%
    \\
    \begin{subfigure}[b]{0.45\textwidth}
        \centering
        \includegraphics[width=\textwidth]{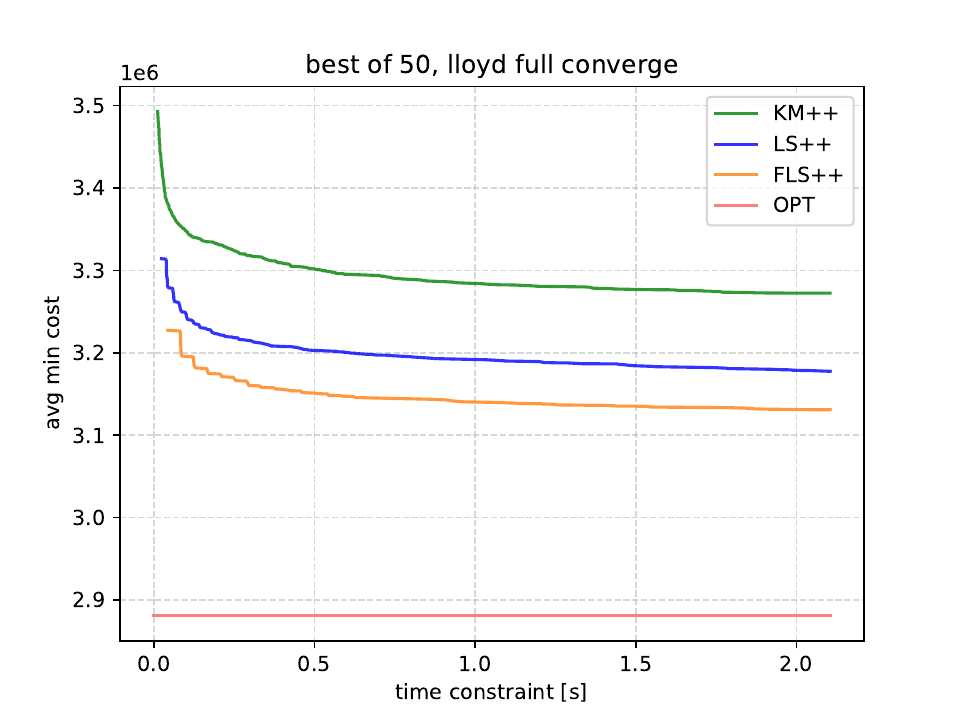}
        \caption{\textit{concrete}, $k=60$}
    \end{subfigure}%
    \begin{subfigure}[b]{0.45\textwidth}
        \centering
        \includegraphics[width=\textwidth]{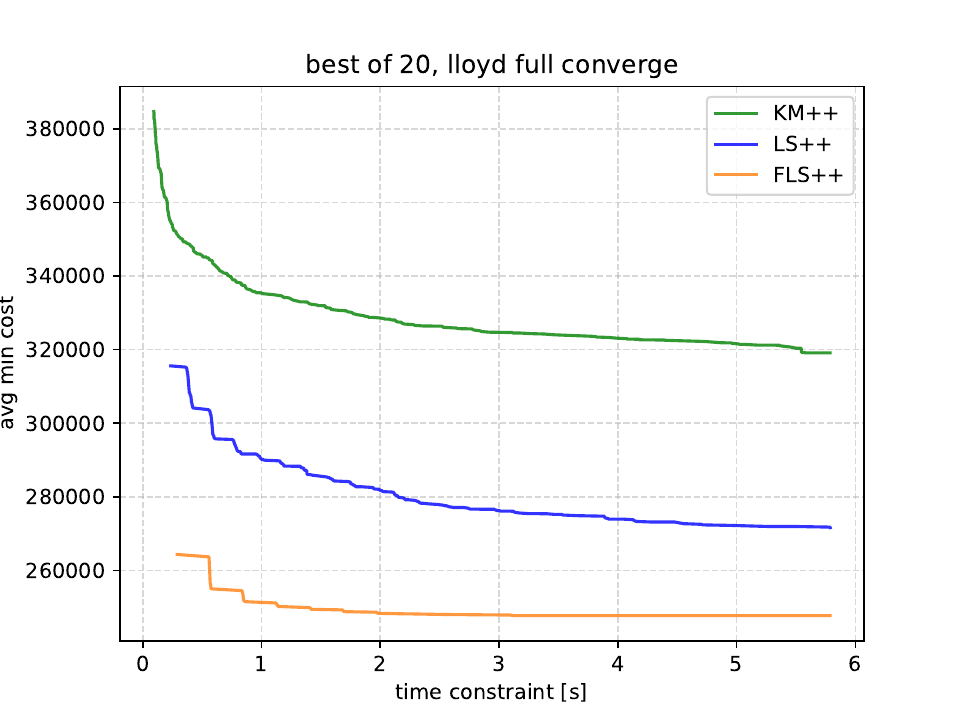}
        \caption{\textit{circles}, $k=100$}
    \end{subfigure}%
    \caption{Comparison of average cost decrease depending on the runtime of FLS++.}
\label{fig:algComp}
\end{figure}

\subparagraph*{Number of best solutions found}
Table~\ref{tab:normal_iterations_wins} shows how many times KM++, LS++ or FLS++ computed the overall best solution for a given dataset within the time limit; most of the time FLS++ did return the smallest solution.
Despite FLS++ having the fewest average number of iterations due to its longer running time, we can see that most of the time it did report the best solution.
In the remaining cases, the best solution was mostly found by LS++, while
KM++, having the most number of average iterations for each dataset, could only report the smallest solution twice, namely in the large dataset \textit{bio train features}.

\begin{table}
\centering
\scalebox{.9625}{
\begin{adjustbox}{width=\textwidth}
\begin{tabular}{|c||c|c||c|c||c|c|}
\hline
data set & \#{}Wins KM++ & \#{}Avg. iterations KM++ & \#{}Wins LS++ & \#{}Avg. iterations LS++ & \#{}Wins FLS++ & \#{}Avg iterations FLS++\\
\hline \hline
\textit{ pr91 } & $5$ & $160.49$ & $8$ & $64.42$ & $87$ & $50$\\ \hline 
\textit{ bio train features } & $2$ & $19.3$ & $6$ & $13.7$ & $12$ & $10$\\ \hline 
\textit{ concrete } & $0$ & $357.02$ & $4$ & $103.22$ & $96$ & $50$\\ \hline 
\textit{ circles } & $0$ & $103.73$ & $0$ & $29.13$ & $94$ & $20$\\ \hline 
\end{tabular}
\end{adjustbox}
}
\caption{Respective number of times each algorithm returned the smallest solution for the respective image in Figure~\ref{fig:algComp} and how many iterations we could repeat the algorithm until the time limit was reached. We only count a win if the respective algorithm returned the strictly smallest cost compared to the other two algorithms in this run.}
\label{tab:normal_iterations_wins}
\end{table}

\subsection{Greedy vs. non-greedy $d^2$-sampling}\label{sec:fancy_greedy_compare}
Although greedy $d^2$-sampling yields a bad theoretical approximation ratio \cite{BERM19}, it is commonly used in practice, e.g. it is the default in the scikit-learn-library with parameter $\ell=2+\lfloor \ln(k) \rfloor$.
In this section, we investigate whether the practical choice is justified 
and how LS++/FLS++ changes when using greedy sampling.

In Figure~\ref{fig:algCompGreedy} we compare the average performance of all algorithms when using \textit{greedy $d^2$-sampling} vs. \textit{standard $d^2$-sampling}.
Here, on almost all evaluated data sets, the average cost decrease is larger when using \textit{greedy $d^2$-sampling} 
regardless of which algorithm is chosen afterwards.
Especially KM++ benefits from using \textit{greedy $d^2$-sampling}.
LS++ and FLS++ benefit from \textit{greedy $d^2$-sampling} most of the time, and by roughly the same amount. 
Using \textit{greedy $d^2$-sampling} on dataset \textit{concrete} in (c) brings each algorithm closer to the optimal solution but a gap remains.
An analysis of the greedy initialization presented in the same way as in Section~\ref{sec:fancy_compare} as well as an analysis how the initialization process did improve the average cost of the returned solution can be found in section~\ref{sec:greedy_init} in the Appendix.
\begin{figure}[!t]
    \begin{subfigure}[b]{0.45\textwidth}
        \centering
        \includegraphics[width=\textwidth]{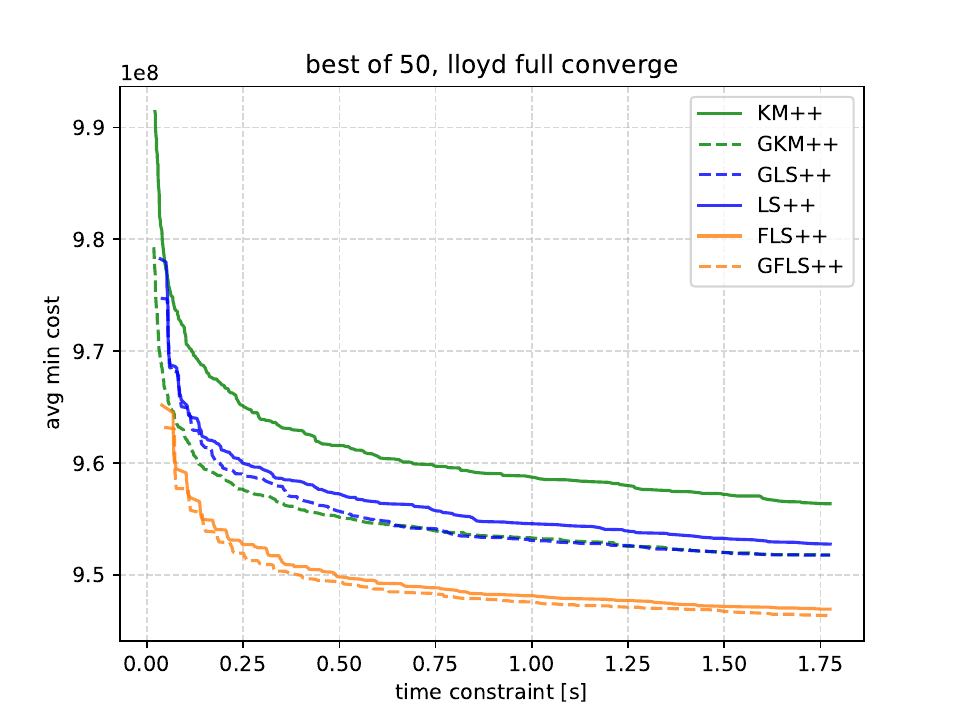}
        \caption{\textit{pr91}, $k=50$}
        \label{fig:gfls_pr91}
    \end{subfigure}%
    \begin{subfigure}[b]{0.45\textwidth}
        \centering
        \includegraphics[width=\textwidth]{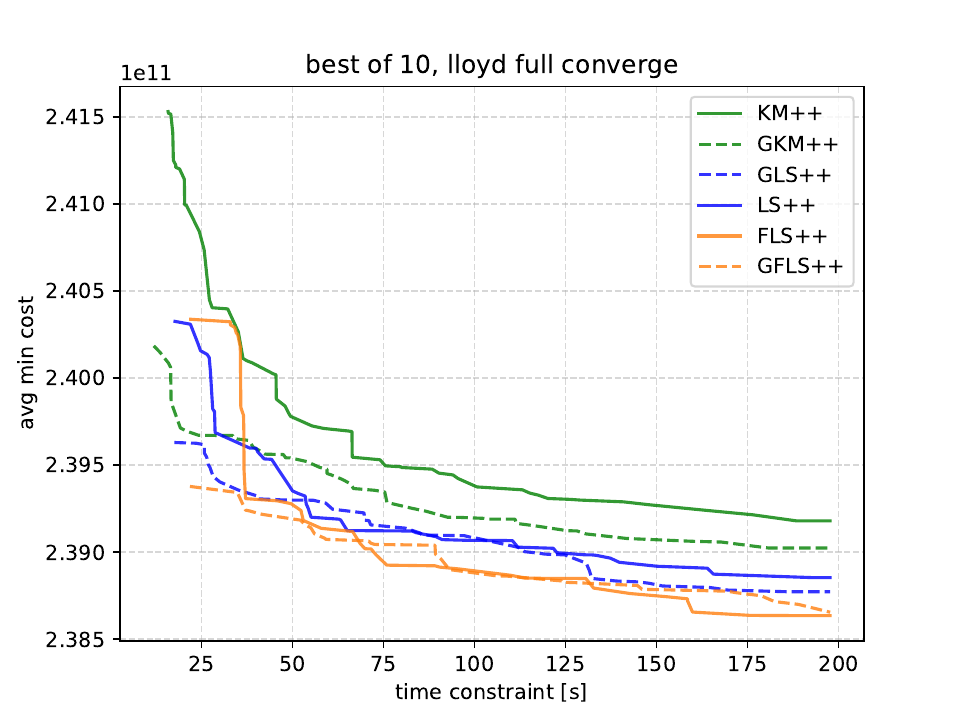}
        \caption{\textit{bio train features}, $k=25$}
    \end{subfigure}%
    \\ 
    \begin{subfigure}[b]{0.45\textwidth}
        \centering
        \includegraphics[width=\textwidth]{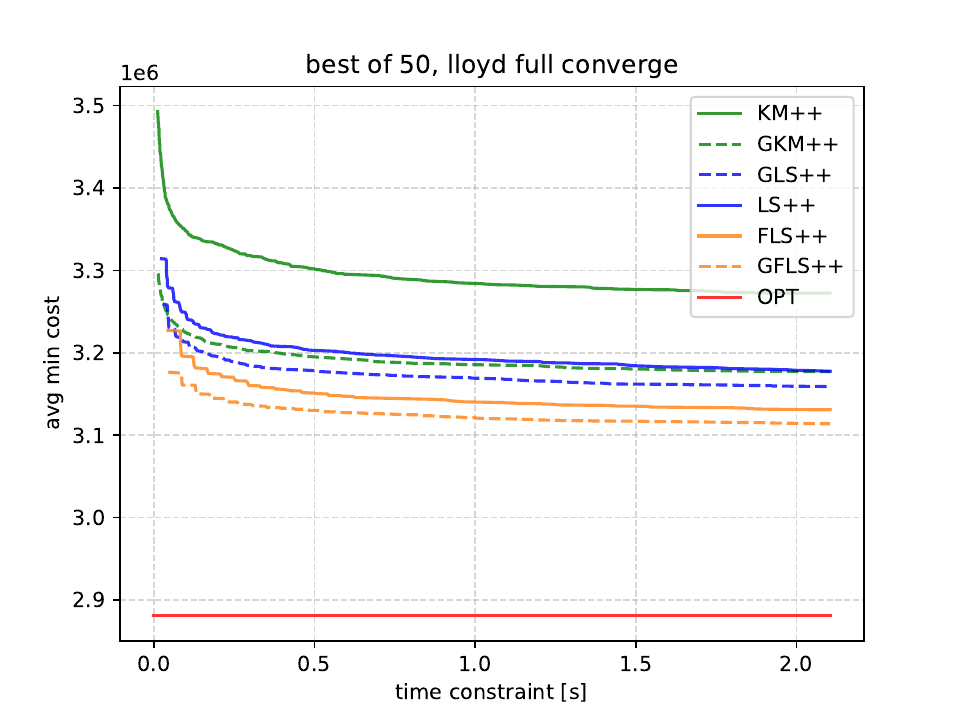}
        \caption{\textit{concrete}, $k=60$}
    \end{subfigure}%
    \begin{subfigure}[b]{0.45\textwidth}
        \centering
        \includegraphics[width=\textwidth]{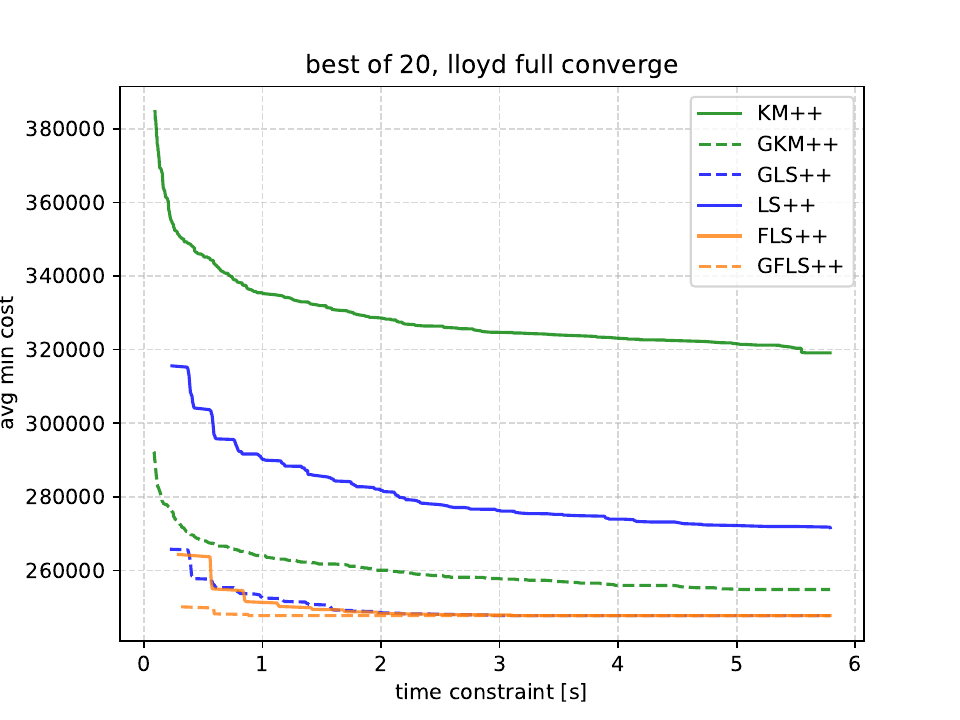}
        \caption{\textit{circles}, $k=100$}
    \end{subfigure}%
    \caption{Comparison of average cost decrease depending on the runtime of FLS++. Dotted lines corresponds to the greedy variant of the original algorithm.}
\label{fig:algCompGreedy}
\end{figure}

\subparagraph*{Varying the number of local search steps}
Running LS++ and FLS++ with the same number of local search steps expectedly leads to FLS++ performing better in terms of cost, but worse in terms of runtime. But Figure~\ref{fig:boxplot} already hints at the possibility that using fewer local search steps in FLS++ might lead to better cost \emph{and} better runtime. 
To this end, we fix the number of local search iterations of LS++ to $25$ (guided by the experiments in~\cite{LS19}), while trying out different values for FLS++. 

More precisely, we let the algorithms perform $50$ runs in total, where one run consists of the following calls. First, we call LS++ with $Z=25$, and afterwards FLS++ for all $Z \in \{5, \ldots, 20\}$. Importantly, we again ensure that each call in the run starts with the same set of initial centers.
After completing all $50$ runs, we compute the average cost and runtime of the LS++ calls across all runs, and similarly the averages for all FLS++ calls. The results are shown in Figure~\ref{fig:var_Z}. The red line indicates the average cost (resp. time) of LS++ with $Z = 25$ over all $50$ runs. Each point on the blue curve corresponds to the average cost (resp. time) of $50$ calls to FLS++ with a certain number of local search iterations.
We see that for FLS++ even very small values of $Z$ already suffice to beat LS++ in terms of average cost. At the same time, these small $Z$ lead to a shorter runtime, making it possible to beat LS++ both in terms of cost and time simultaneously.

To ensure that LS++ could not theoretically improve by, given some time limit, rerunning it multiple times we use a similar setup as in Section~\ref{sec:fancy_compare} to analyze the average performance for some given time bound
in the Appendix, see Figure~\ref{fig:same_ls_steps}. 
In this case we use the time limit given by FLS++ for $B=25$ for all algorithms.
We also fixed $Z=25$ for LS++ and GLS++ but varied the number of local search steps for FLS++ and GFLS++ as $Z\in\{5,10,15,20\}$.
Even with fewer local search steps in FLS++ or GLS++ we get on average a smaller cost compared to LS++ or GLS++.
For two considered datasets \textit{concrete} and \textit{pr91}, choosing $Z=15$ was always enough to get approximately the same cost. For \textit{pr91}, even choosing $Z=5$ did result in a better solution than LS++ or GLS++ on average for the specified time limit.

\begin{figure}[!hb]
    \begin{subfigure}[b]{0.4\textwidth}
        \centering
        \includegraphics[width=\textwidth]{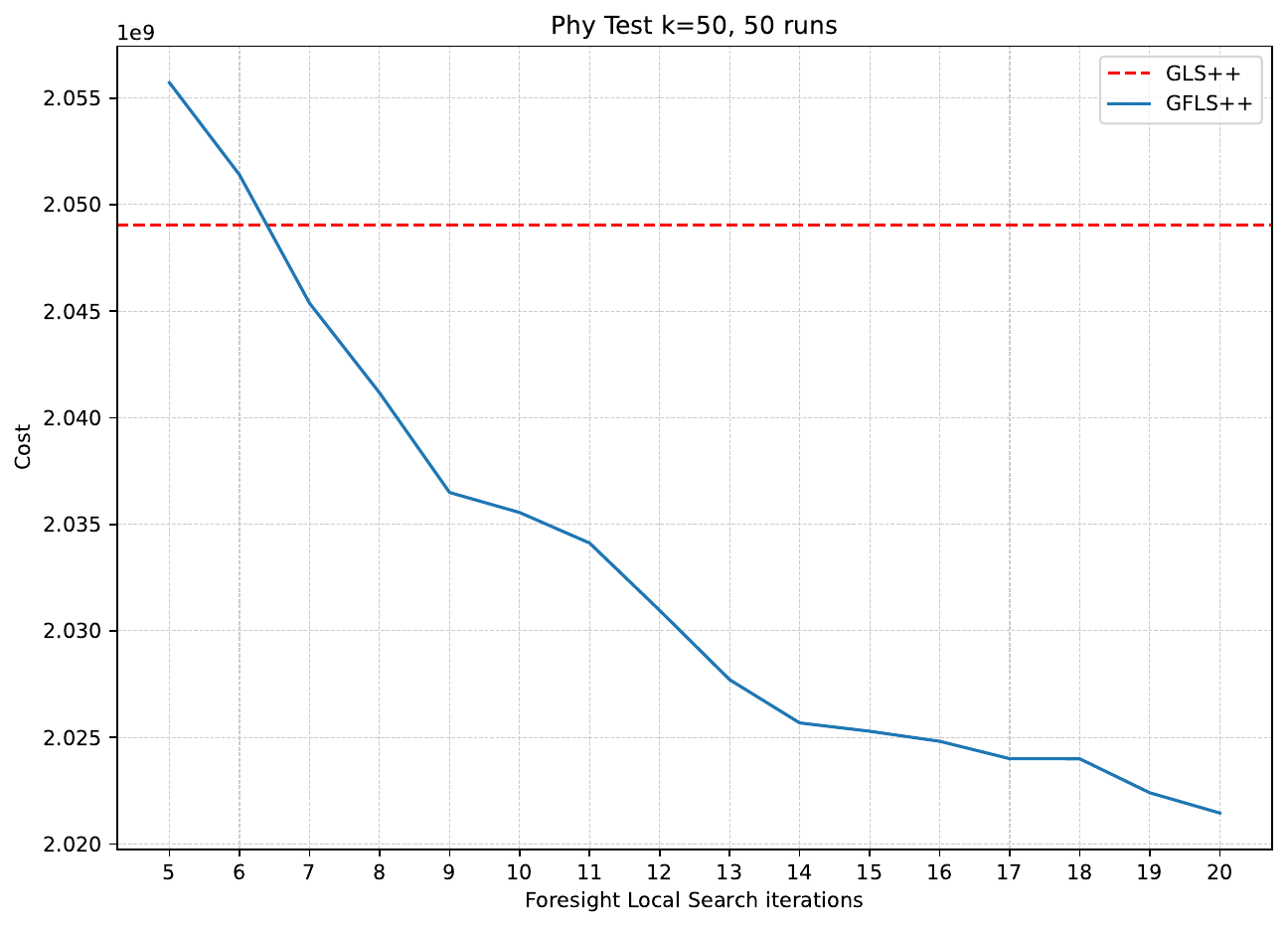}
        \caption{\textit{KDD Phy Test} Cost, $k=50$}
    \end{subfigure}%
    \begin{subfigure}[b]{0.4\textwidth}
        \centering
        \includegraphics[width=\textwidth]{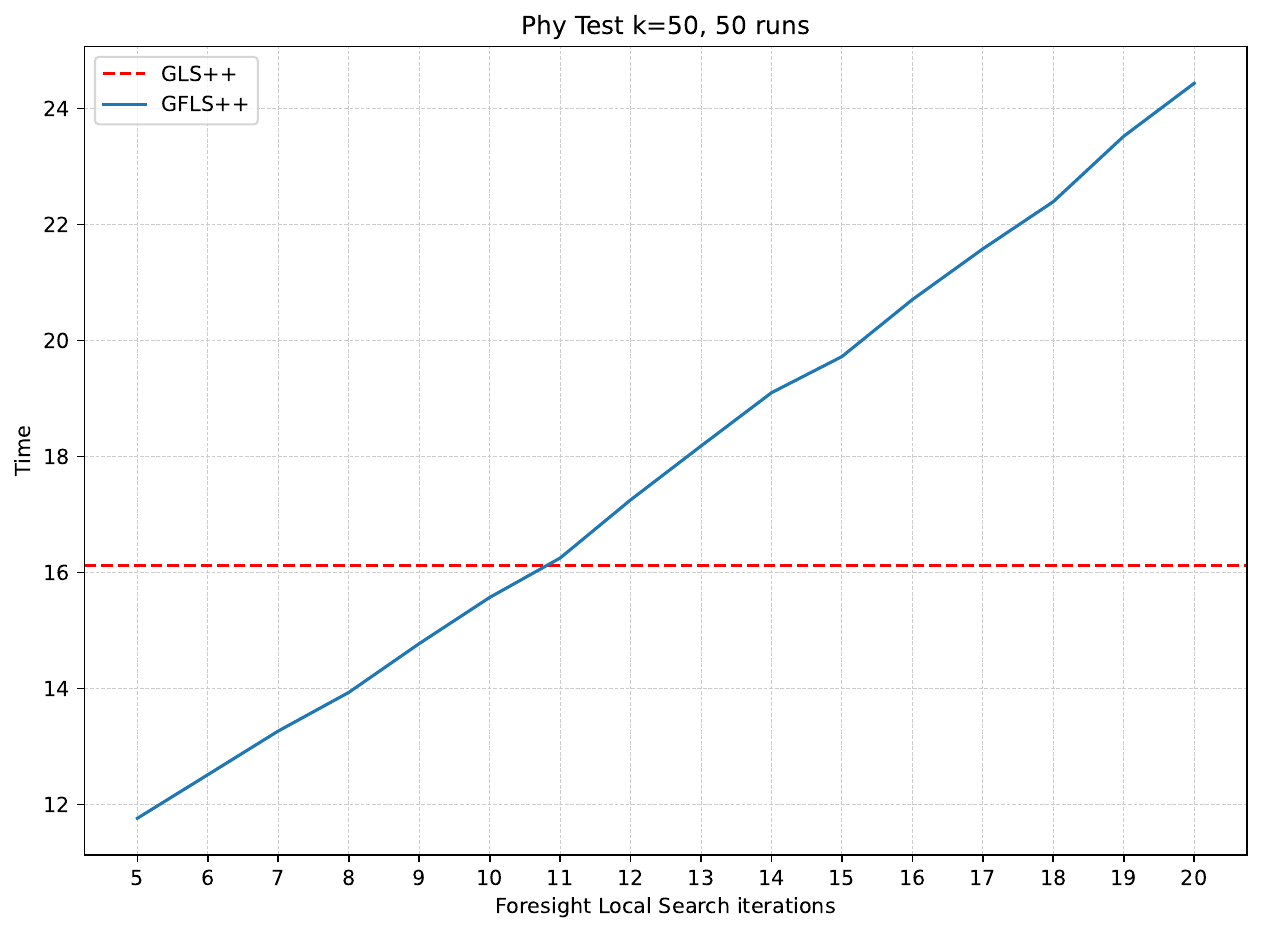}
        \caption{\textit{KDD Phy Test} Time, $k=50$}
    \end{subfigure}%
    \\
    \begin{subfigure}[b]{0.4\textwidth}
        \centering
        \includegraphics[width=\textwidth]{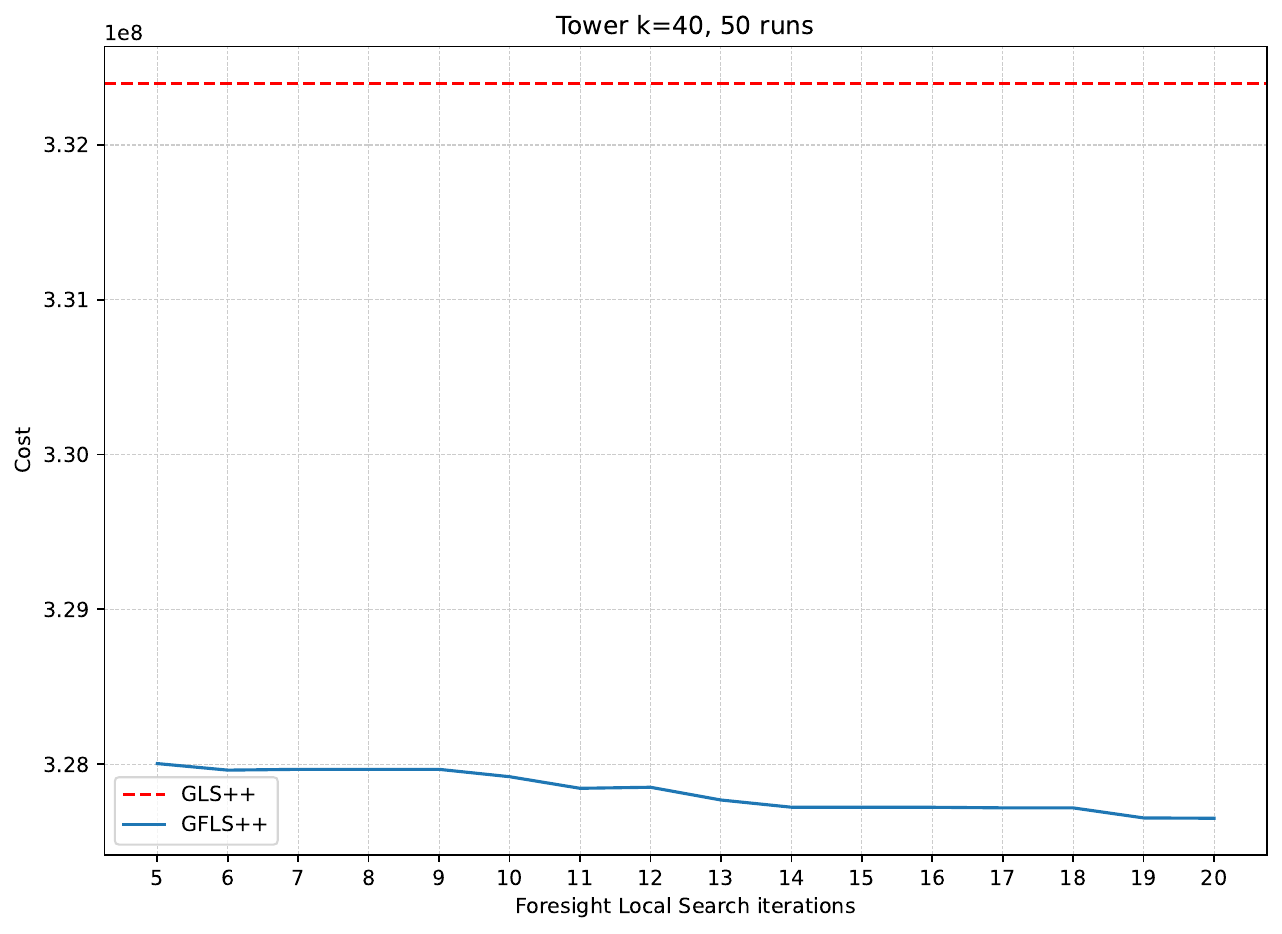}
        \caption{\textit{Tower} Cost, $k=40$}
    \end{subfigure}%
    \begin{subfigure}[b]{0.4\textwidth}
        \centering
        \includegraphics[width=\textwidth]{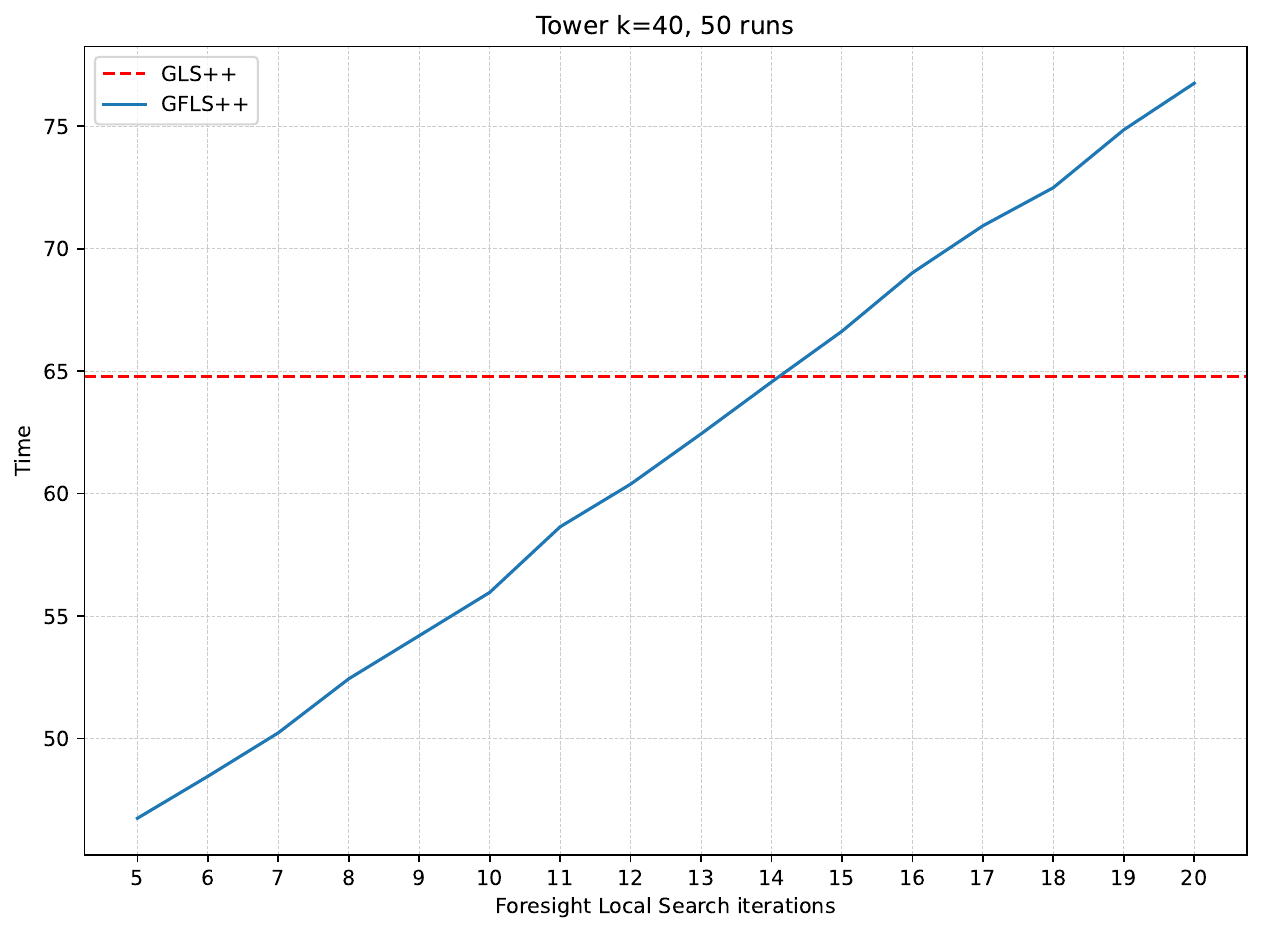}
        \caption{\textit{Tower} Time, $k=50$}
    \end{subfigure}%
    \caption{Impact on cost and time of increasing the number of local search steps in FLS++. The dashed red line shows average cost (resp. time) of $50$ runs of LS++ with $25$ local search steps.}
\label{fig:var_Z}
\end{figure}

\section{Conclusion}
We propose the new algorithm FLS++ for the \textit{k-means} problem: The algorithm uses \emph{foresight} by combining local search with $d^2$-sampling and can outperform the established methods of $k$-means++ and LS++ in terms of both quality and running time. 
Even though FLS++ only differs slightly from LS++, our experiments suggest that performing just one Lloyd iteration in each local search step often leads to better performance in terms of cost as well as runtime. 
Additionally, we investigate if the popularity of \textit{greedy $d^2$-initialization} in practice is justified 
even though it performs poorly in theory. Surprisingly, it turns out that on most data sets considered here, all standard algorithms (as well as our new algorithm) perform substantially better when using greedy initialization.

On the other hand, running our algorithm \emph{without} greedy initialization offers a robust way to recover the quality of the greedy initialization without sacrificing running time.

%% file: appendix.tex
\appendix
\section*{Appendix}

\section{Further datasets and cumulative costs}
\begin{figure}[b]
    \centering
    \begin{subfigure}[b]{.45\textwidth}
        \includegraphics[width=\linewidth]{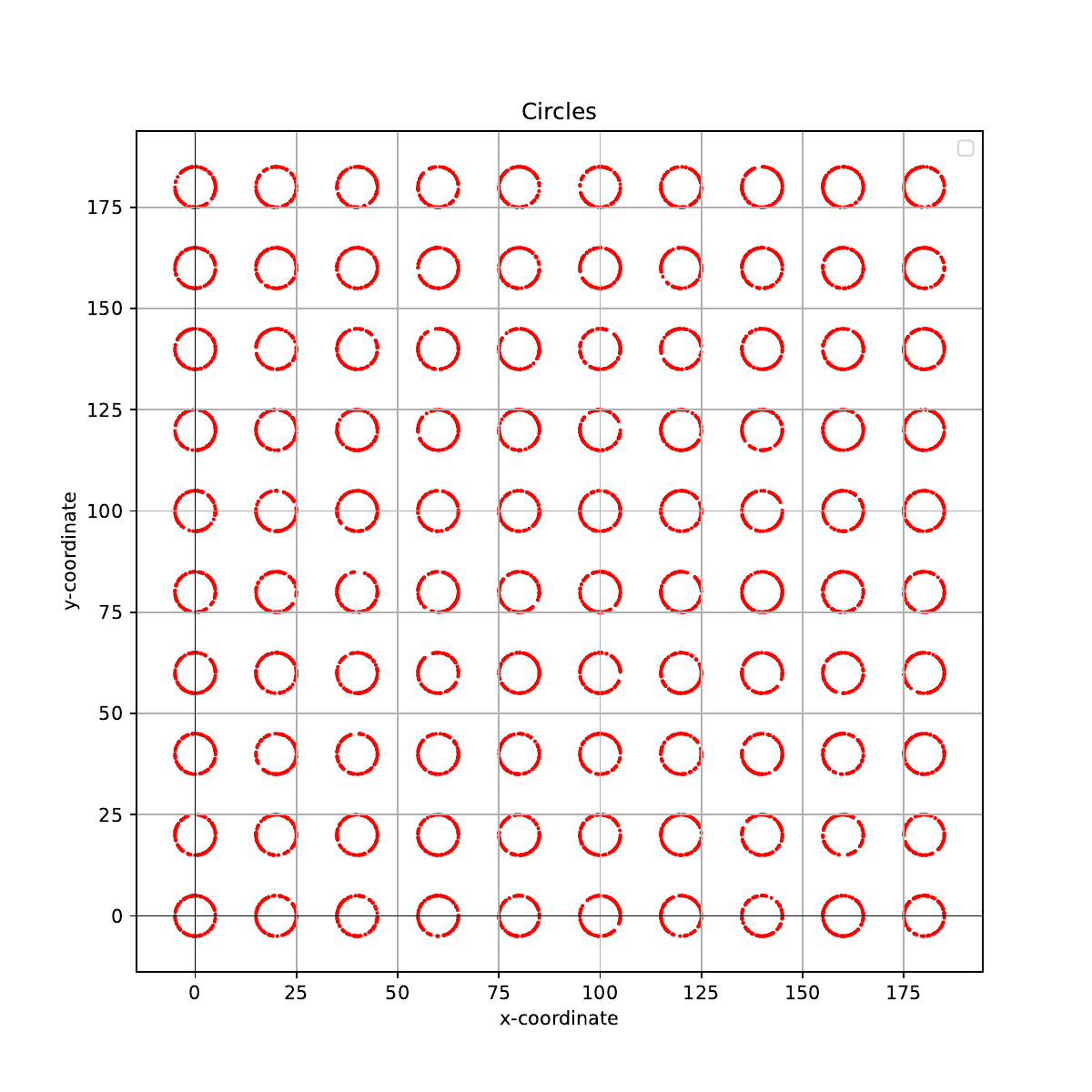}
        \caption{Circles dataset}
    \end{subfigure}
    \hfill
    \begin{subfigure}[b]{.45\textwidth}
        \includegraphics[width=\linewidth]{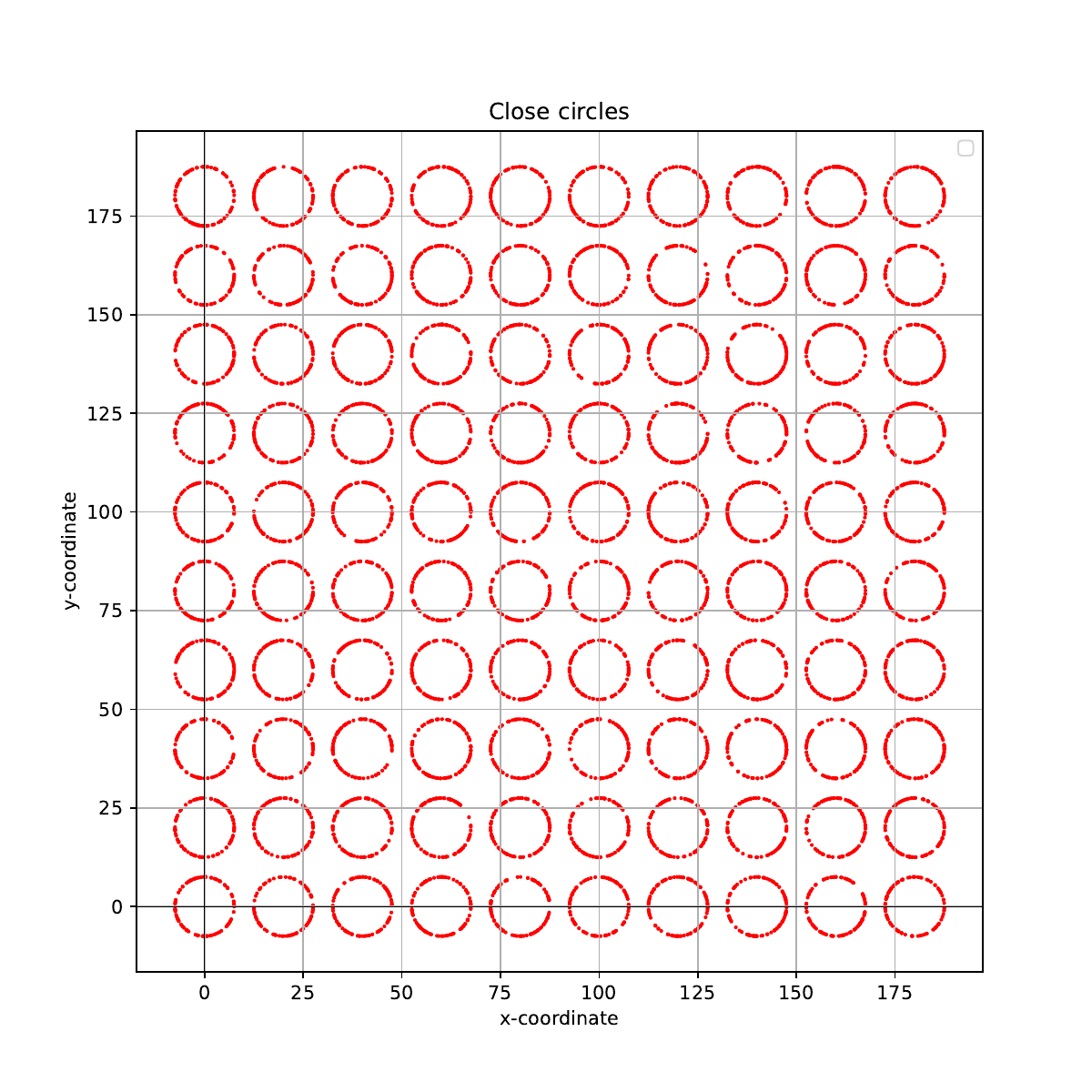}
        \caption{Close circles dataset}
    \end{subfigure}
    \caption{Two artificially generated datasets also used in the experiments.}
    \label{fig:circle_circle}
\end{figure}

\begin{table}[b]
\centering
\scalebox{.9625}{
\begin{adjustbox}{width=\textwidth}
\begin{tabular}{|c||c|c||c|c||c|c|}
\hline
data set & \#{}Wins KM++ & \#{}Avg. iterations KM++ & \#{}Wins LS++ & \#{}Avg. iterations LS++ & \#{}Wins FLS++ & \#{}Avg iterations FLS++\\
\hline \hline
\textit{ close circles } & $0$ & $134.27$ & $0$ & $57.68$ & $100$ & $50$\\ \hline  
\textit{ D31 } & $0$ & $229.98$ & $2$ & $73.55$ & $3$ & $50$\\ \hline 
\textit{ frymire } & $0$ & $42.6$ & $5$ & $14.05$ & $15$ & $10$\\ \hline 
\textit{ rectangles } & $0$ & $453.67$ & $0$ & $88.66$ & $0$ & $50$\\ \hline 
\textit{ s3 } & $0$ & $127.15$ & $3$ & $58.58$ & $97$ & $50$\\ \hline 
\textit{ Body measurements } & $0$ & $401.23$ & $0$ & $108.15$ & $100$ & $50$\\ \hline 
\end{tabular}
\end{adjustbox}
}
\caption{Respective number of times each algorithm returned the smallest solution for the respective image in Figure~\ref{fig:algCompGreedy2} and how many iterations we could repeat the algorithm until the time limit was reached. We only count a win if the respective algorithm returned the strictly smallest cost compared to the other two algorithms in this run.}
\label{tab:normal_iterations_wins2}
\end{table}

\begin{table}[H]
\centering
\scalebox{.9625}{
\begin{adjustbox}{width=\textwidth}
\begin{tabular}{|c||c|c||c|c||c|c|}
\hline
data set & \#{}Wins GKM++ & \#{}Avg. iterations GKM++ & \#{}Wins GLS++ & \#{}Avg. iterations GLS++ & \#{}Wins GFLS++ & \#{}Avg. iterations GFLS++\\
\hline \hline
\textit{ pr91 } & $5$ & $160.39$ & $8$ & $62.4$ & $87$ & $48.14$\\ \hline 
\textit{ bio train features } & $2$ & $20.9$ & $6$ & $12.95$ & $12$ & $10.1$\\ \hline 
\textit{ concrete } & $0$ & $295.71$ & $0$ & $95.44$ & $100$ & $48.25$\\ \hline 
\textit{ circles } & $0$ & $109.96$ & $0$ & $29.33$ & $0$ & $19.86$\\ \hline 
\hline
\textit{ close circles } & $3$ & $127.14$ & $2$ & $55.47$ & $95$ & $48.46$\\ \hline 
\textit{ D31 } & $0$ & $250.55$ & $1$ & $69.75$ & $0$ & $47.3$\\ \hline 
\textit{ frymire } & $5$ & $41.15$ & $2$ & $13.55$ & $13$ & $10$\\ \hline 
\textit{ rectangles } & $0$ & $442.31$ & $0$ & $83.12$ & $0$ & $48.12$\\ \hline 
\textit{ s3 } & $4$ & $131$ & $6$ & $57.38$ & $89$ & $48.38$\\ \hline 
\textit{ Body measurements } & $0$ & $336.77$ & $0$ & $101.81$ & $100$ & $47.72$\\ \hline 
\end{tabular}
\end{adjustbox}
}
\caption{Respective number of times each algorithm returned the smallest solution for the respective image in Figure~\ref{fig:algCompGreedy} and how many iterations we could repeat the algorithm until the time limit was reached. We only count a win if the respective algorithm returned the strictly smallest cost compared to the other two algorithms in this run.}
\label{tab:greedy_iterations_wins}
\end{table}

In the following we also show the performance of each algorithm with \textit{normal $d^2$-sampling} and \textit{greedy $d^2$-sampling} for some other datasets.
Like in sections~\ref{sec:fancy_compare} and \ref{sec:fancy_greedy_compare} we evaluate the performances of each algorithm using only \textit{normal $d^2$-sampling} and then with \textit{greedy $d^2$-sampling} by checking f.e. which algorithm did return the best solution in every run or the cost factors between the algorithms.
Figure~\ref{fig:algCompGreedy2} shows the development of the average found minimum cost for each algorithm for 6 other datasets.
Each dataset was averaged over $100$ iterations, except for the dataset \textit{frymire}, which is averaged over $20$ iterations.

We can see in Figure~\ref{fig:algCompGreedy2} almost the same image as for the previous cases in section~\ref{sec:fancy_compare}. For datasets \textit{D31} and \textit{rectangles} we additionally see that all algorithms except for KM++ almost never fail to find the optimal solution in a short amount of time.

\begin{figure}
    \begin{subfigure}[b]{0.32\textwidth}
        \centering
        \includegraphics[width=\textwidth]{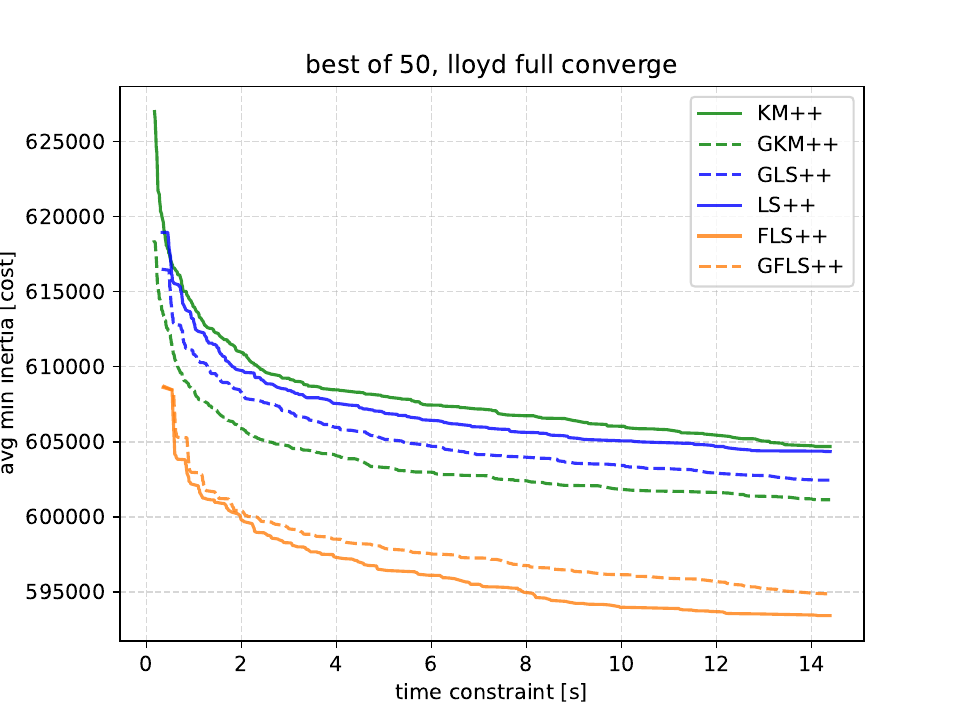}
        \caption{\textit{close circles}, $k=100$}
        \label{fig:closecircles}
    \end{subfigure}%
    \begin{subfigure}[b]{0.32\textwidth}
        \centering
        \includegraphics[width=\textwidth]{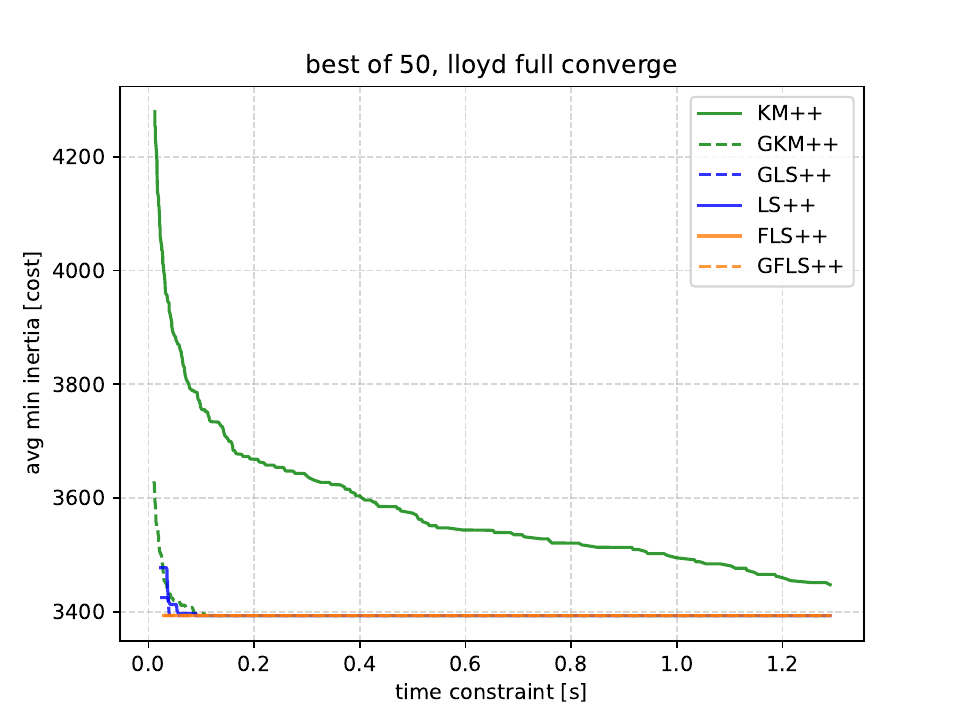}
        \caption{\textit{D31}, $k=31$}
    \end{subfigure}%
    \begin{subfigure}[b]{0.32\textwidth}
        \centering
        \includegraphics[width=\textwidth]{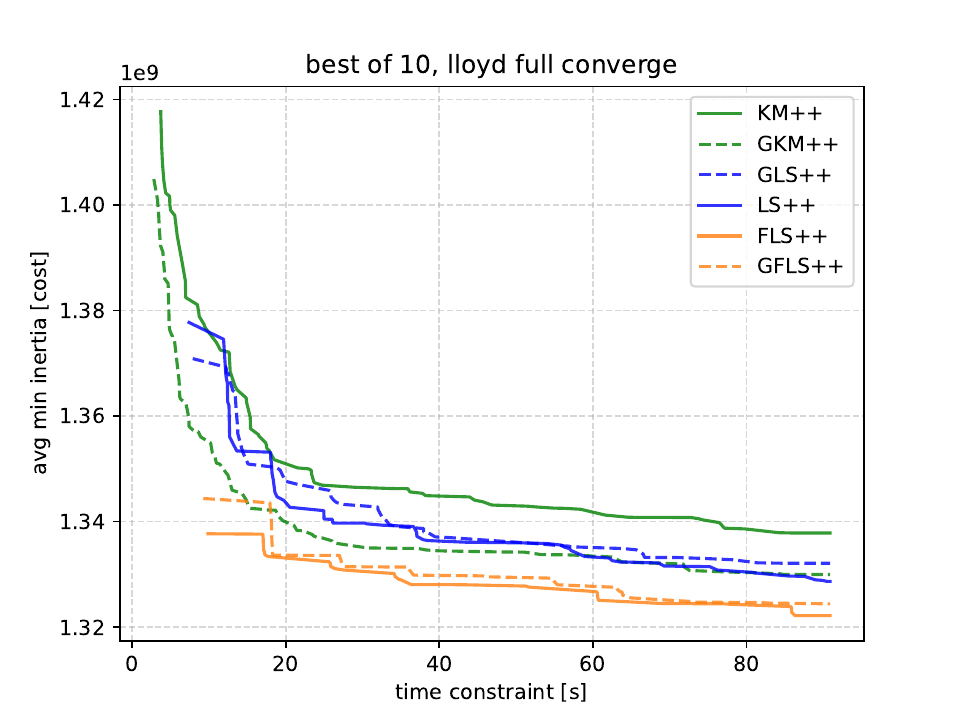}
        \caption{\textit{frymire}, $k=20$}
    \end{subfigure}%
    \\
    \begin{subfigure}[b]{0.32\textwidth}
        \centering
        \includegraphics[width=\textwidth]{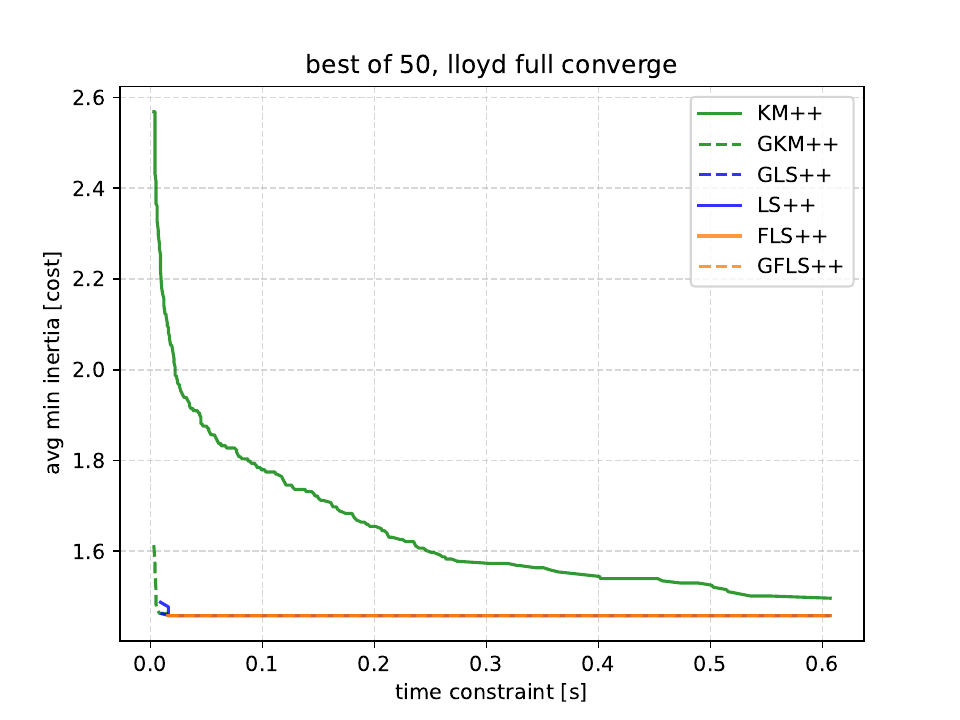}
        \caption{\textit{rectangles}, $k=36$}
    \end{subfigure}
    \begin{subfigure}[b]{0.32\textwidth}
        \centering
        \includegraphics[width=\textwidth]{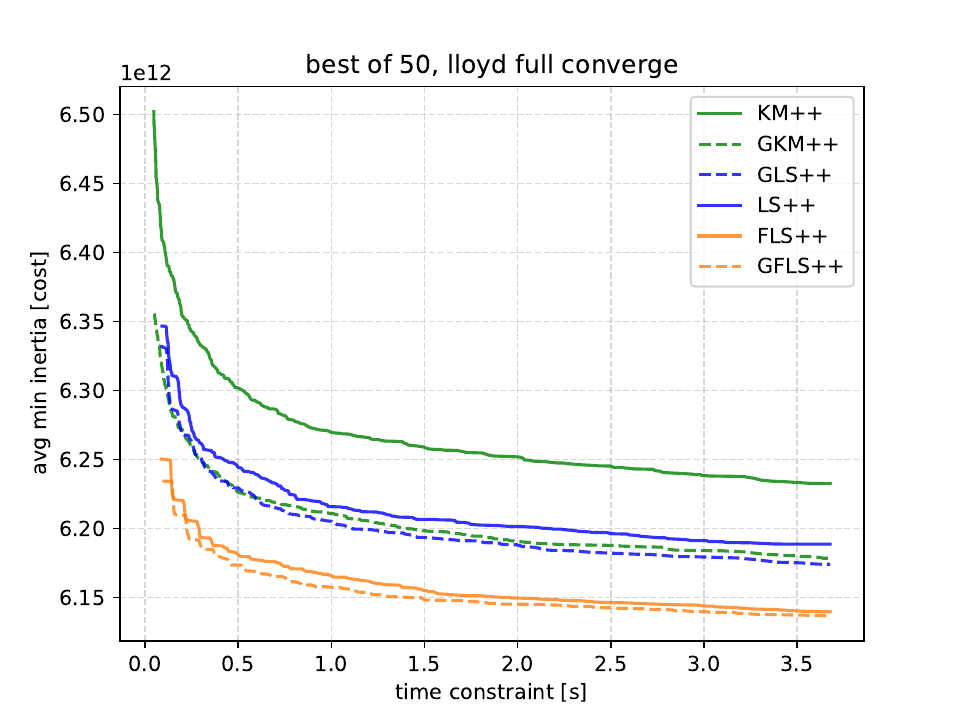}
        \caption{\textit{s3}, $k=50$}
    \end{subfigure}
    \begin{subfigure}[b]{0.32\textwidth}
        \centering
        \includegraphics[width=\textwidth]{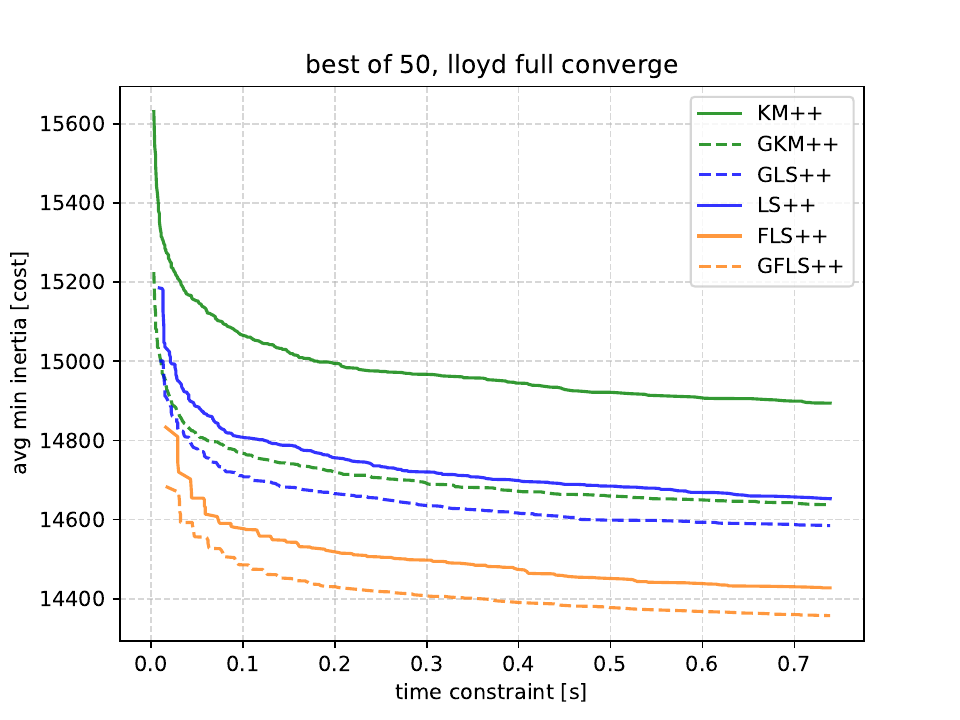}
        \caption{\textit{Body measurements}, $k=50$}
    \end{subfigure}
    \caption{Comparison of average cost decrease depending on the runtime of FLS++. Dotted lines corresponds to the greedy variant of the original algorithm.}
\label{fig:algCompGreedy2}
\end{figure}

\begin{table}
\centering
\scalebox{.9625}{
\begin{adjustbox}{width=\textwidth}
\begin{tabular}{|c||Sc|Sc|Sc|c|c|}
\hline
data set & $c_{\text{KM++}}$ & $c_{\text{LS++}}$ & $c_{\text{FLS++}}$ & $C(\text{FLS++}, \text{KM++})$ & $C(\text{FLS++}, \text{LS++})$\\
\hline \hline
\textit{ close circles } & $ 6.0468\text{{E+}}05 $ & $ 6.0434\text{{E+}}05 $ & $ 5.9341\text{{E+}}05 $ & $ 1.86\% $ & $ 1.81\% $ \\      \hline
\textit{ D31 } & $ 3447.57 $ & $ 3393.26 $ & $ 3393.26 $ & $ 1.58\% $ & $ 0\% $ \\      \hline
\textit{ frymire } & $ 1.3378\text{{E+}}09 $ & $ 1.3287\text{{E+}}09 $ & $ 1.3222\text{{E+}}09 $ & $ 1.17\% $ & $ 0.49\% $ \\      \hline
\textit{ rectangles } & $ 1.5 $ & $ 1.46 $ & $ 1.46 $ & $ 2.56\% $ & $ 0\% $ \\      \hline
\textit{ s3 } & $ 6.2325\text{{E+}}12 $ & $ 6.1887\text{{E+}}12 $ & $ 6.1397\text{{E+}}12 $ & $ 1.49\% $ & $ 0.79\% $ \\      \hline
\textit{ Body measurements } & $ 14894.11 $ & $ 14652.64 $ & $ 14427.42 $ & $ 3.13\% $ & $ 1.54\% $ \\      \hline
\end{tabular}
\end{adjustbox}
}
\caption{Respective average cost of each algorithm. Last two columns show $(1-c_{\text{FLS++}}/c_A) \cdot 100\%$ for $A\in \{\text{KM++}, \text{LS++}\}$. }
\label{tab:normal_iterations_costs2}
\end{table}

\section{Cost improvement through greedy initialization}\label{sec:greedy_init}

Now we want to consider how using \textit{greedy $d^2$-sampling} did improve the average costs after reaching the time limits in section~\ref{sec:fancy_greedy_compare}.

For $A\in \{\text{FLS++}, \text{LS++}, \text{KM++}\}$ we define their average costs $c_A,c_A^G$ where the superscript G indicates if we use \textit{greedy $d^2$-sampling}.
In our case we compare both average costs over all runs $r\in [R]$ when using time limit $t_B^r$.
Lastly, similar to the analysis above, we define the improvement factor as $(1-\frac{c_A^G}{c_A})\cdot 100\%$.
As we can see in Table~\ref{tab:greedy_nogreedy_comp_table} for most data sets and algorithms the improvement is not below $0\%$ and most of the time larger than $0\%$.

\begin{table}
\centering
\scalebox{.8}{
\begin{adjustbox}{width=\textwidth}
\begin{tabular}{|c||c|c|c|}
\hline
data set & $C(\text{GKM++}, \text{KM++})$ & $C(\text{GLS++}, \text{LS++})$ & $C(\text{GFLS++}, \text{FLS++})$ \\
\hline \hline
\textit{ pr91 } & $0.48\%$ & $0.1\%$ & $0.06\% $\\ \hline 
\textit{ bio train features } & $0.07\%$ & $0.03\%$ & $-0.01\% $\\ \hline 
\textit{ concrete } & $2.91\%$ & $0.58\%$ & $0.54\% $\\ \hline 
\textit{ circles } & $20.16\%$ & $8.8\%$ & $0\% $\\ \hline 
\hline 
\textit{ close circles } & $0.59\%$ & $0.32\%$ & $-0.24\% $\\ \hline 
\textit{ D31 } & $1.58\%$ & $0\%$ & $0\% $\\ \hline 
\textit{ frymire } & $0.59\%$ & $-0.26\%$ & $-0.17\% $\\ \hline 
\textit{ rectangles } & $2.56\%$ & $0\%$ & $0\% $\\ \hline 
\textit{ s3 } & $0.87\%$ & $0.23\%$ & $0.04\% $\\ \hline 
\textit{ Body measurements } & $1.72\%$ & $0.46\%$ & $0.49\% $\\ \hline 
\end{tabular}
\end{adjustbox}
}
\caption{Average cost decrease by using \textit{greedy $d^2$-sampling}.}
\label{tab:greedy_nogreedy_comp_table}
\end{table}

At last we analyse the number of times each algorihm did report the strictly best solution, the number of iterations, average costs and cost factor when compared to GFLS++.
Like in the case when using \textit{normal $d^2$-sampling} we can see in Tables~\ref{tab:greedy_iterations_wins} and~\ref{tab:greedy_iterations_costs} that GFLS++ still on average computes the solution with the smallest cost and does so with large probability. In case of the dataset \textit{circles}, now GLS++ and GFLS++ computed the optimal solution after the timelimit in every round. But we can see in Figure~\ref{fig:algCompGreedy} that GFLS++ does so faster than GLS++.

\begin{table}
\centering
\scalebox{.9625}{
\begin{adjustbox}{width=\textwidth}
\begin{tabular}{|c||Sc|Sc|Sc|c|c|}
\hline
data set & $c_{\text{GKM++}}$ & $c_{\text{GLS++}}$ & $c_{\text{GFLS++}}$ & $C(\text{GFLS++}, \text{GKM++})$ & $C(\text{GFLS++}, \text{GLS++})$\\
\hline \hline
\textit{ pr91 } & $ 9.518\text{{E+}}08 $ & $ 9.5177\text{{E+}}08 $ & $ 9.4638\text{{E+}}08 $ & $ 0.57\% $ & $ 0.57\% $ \\      \hline
\textit{ bio train features } & $ 2.3902\text{{E+}}11 $ & $ 2.3877\text{{E+}}11 $ & $ 2.3866\text{{E+}}11 $ & $ 0.15\% $ & $ 0.05\% $ \\      \hline
\textit{ concrete } & $ 3.1772\text{{E+}}06 $ & $ 3.1591\text{{E+}}06 $ & $ 3.1138\text{{E+}}06 $ & $ 2\% $ & $ 1.43\% $ \\      \hline
\textit{ circles } & $ 2.5481\text{{E+}}05 $ & $ 2.4773\text{{E+}}05 $ & $ 2.4773\text{{E+}}05 $ & $ 2.78\% $ & $ 0\% $ \\      \hline
\hline
\textit{ close circles } & $ 6.0113\text{{E+}}05 $ & $ 6.0244\text{{E+}}05 $ & $ 5.9485\text{{E+}}05 $ & $ 1.04\% $ & $ 1.26\% $ \\      \hline
\textit{ D31 } & $ 3393.26 $ & $ 3393.26 $ & $ 3393.26 $ & $ 0\% $ & $ -0\% $ \\      \hline
\textit{ frymire } & $ 1.33\text{{E+}}09 $ & $ 1.3321\text{{E+}}09 $ & $ 1.3244\text{{E+}}09 $ & $ 0.42\% $ & $ 0.58\% $ \\      \hline
\textit{ rectangles } & $ 1.46 $ & $ 1.46 $ & $ 1.46 $ & $ 0\% $ & $ 0\% $ \\      \hline
\textit{ s3 } & $ 6.1786\text{{E+}}12 $ & $ 6.1742\text{{E+}}12 $ & $ 6.137\text{{E+}}12 $ & $ 0.67\% $ & $ 0.6\% $ \\      \hline
\textit{ Body measurements } & $ 14637.51 $ & $ 14584.68 $ & $ 14357.35 $ & $ 1.91\% $ & $ 1.56\% $ \\      \hline
\end{tabular}
\end{adjustbox}
}
\caption{Respective average cost of each algorithm. Last two columns show $(1-c_{\text{GFLS++}}/c_A) \cdot 100\%$ for $A\in \{\text{GKM++}, \text{GLS++}\}$. }
\label{tab:greedy_iterations_costs}
\end{table}

\section{Local search remaining Figures}\label{sec:remaining_figures}

\begin{figure}
    \begin{subfigure}[b]{0.45\textwidth}
        \centering
        \includegraphics[width=\textwidth]{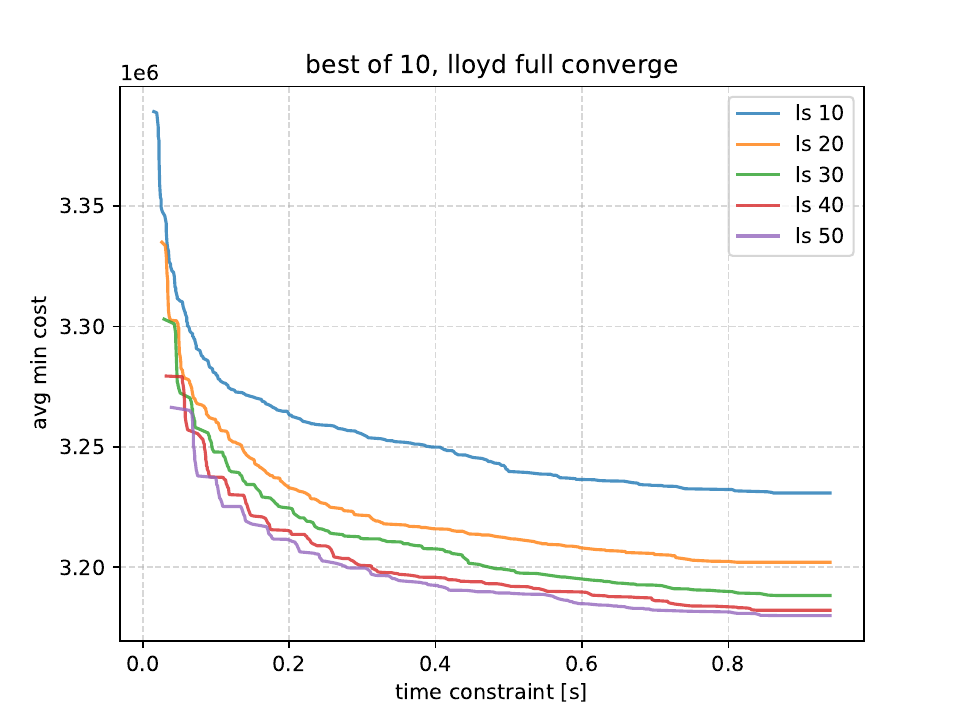}
        \caption{cost developments LS++.}
    \end{subfigure}%
    \begin{subfigure}[b]{0.45\textwidth}
        \centering
        \includegraphics[width=\textwidth]{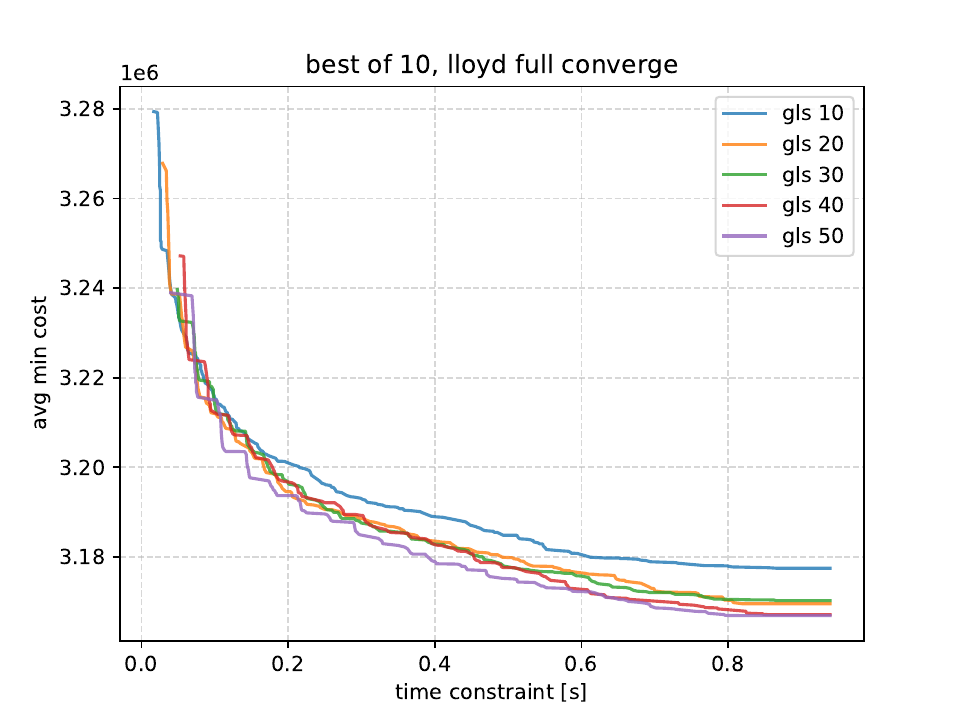}
        \caption{cost developments GLS++.}
    \end{subfigure}%
    \\
    \begin{subfigure}[b]{0.45\textwidth}
        \centering
        \includegraphics[width=\textwidth]{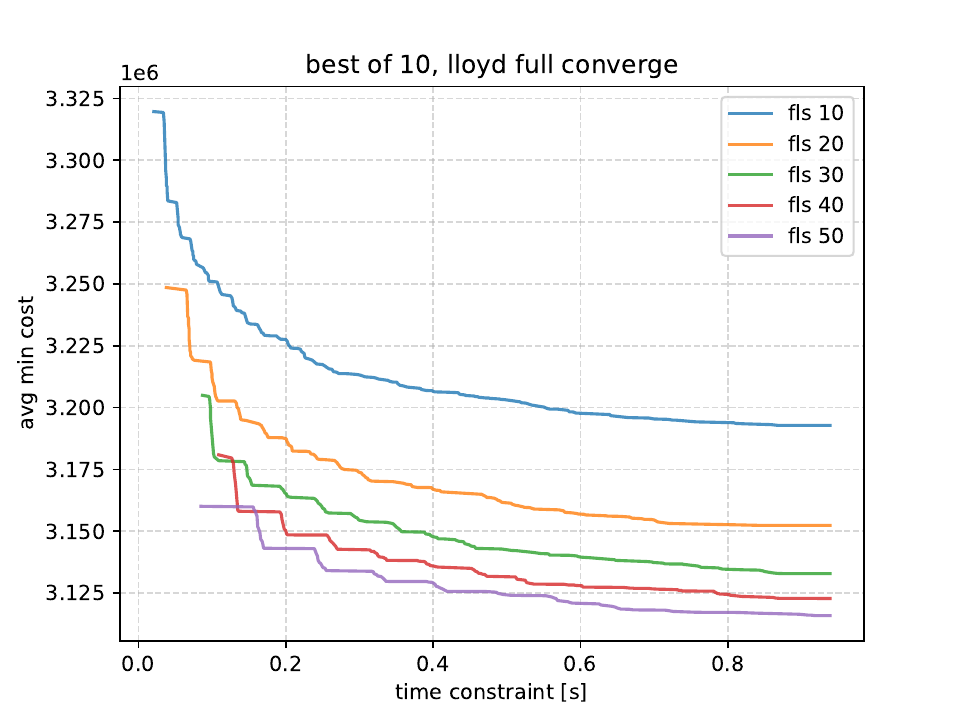}
        \caption{cost developments FLS++.}
    \end{subfigure}%
    \begin{subfigure}[b]{0.45\textwidth}
        \centering
        \includegraphics[width=\textwidth]{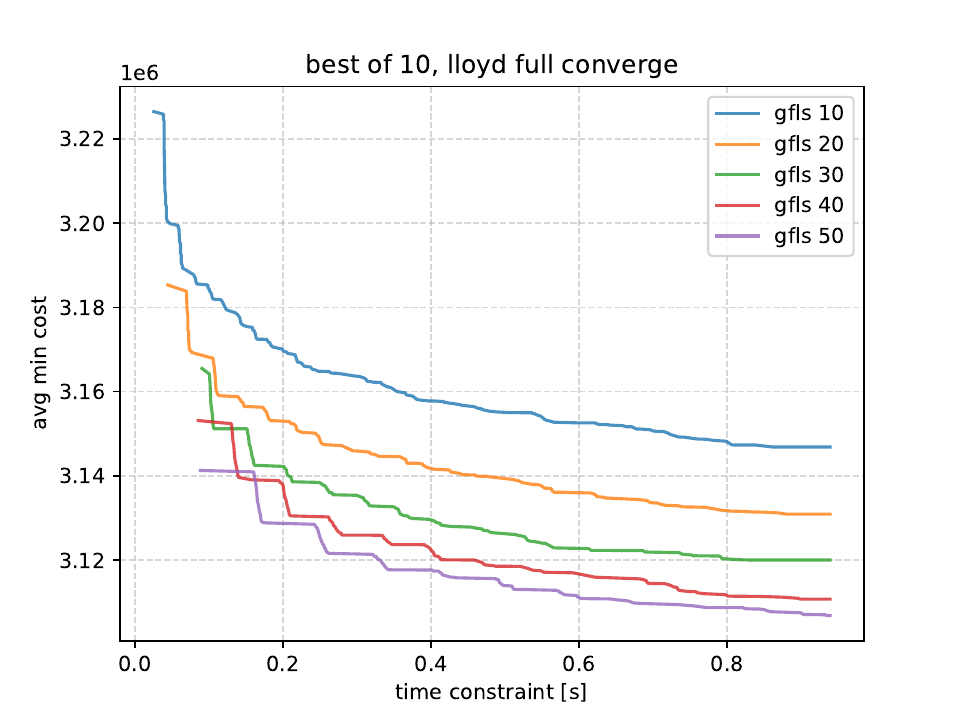}
        \caption{cost developments GFLS++.}
    \end{subfigure}%
    \caption{Dataset \textit{concrete}. Comparison of the average cost decrease for (G)LS++ and (G)FLS++ for $k=60$, $R=100$ and varying values of $Z$.} 
\label{fig:var_Z_fancy}
\end{figure}

\begin{figure}
    \begin{subfigure}[b]{0.45\textwidth}
    \includegraphics[width=\textwidth]{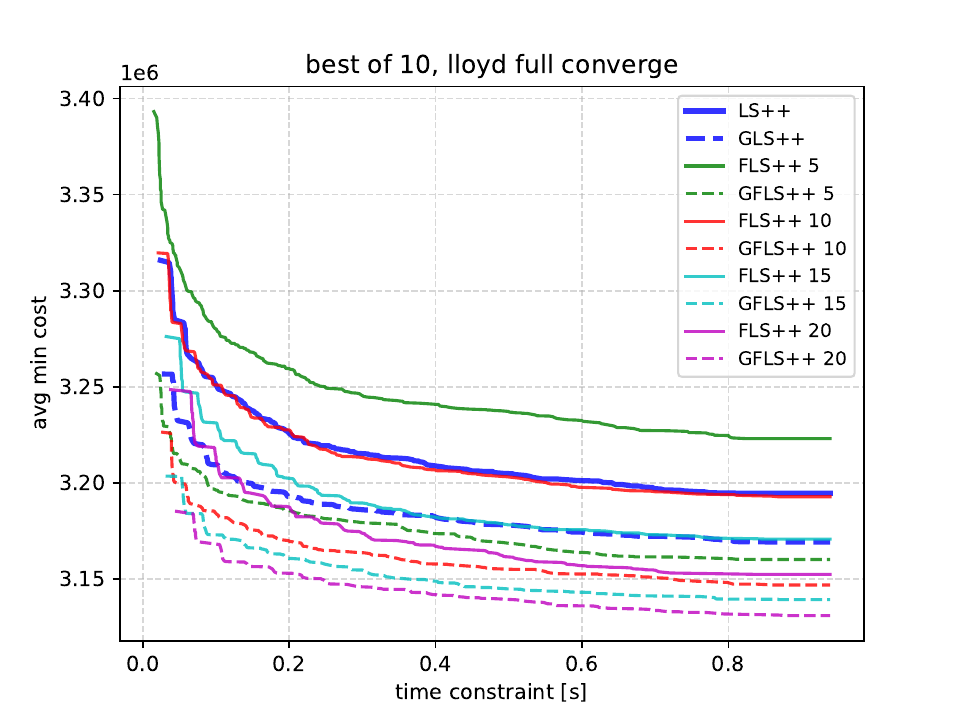}
    \caption{Dataset \textit{concrete}, $k=60$.}
    \end{subfigure}
    \begin{subfigure}[b]{0.45\textwidth}
        \includegraphics[width=\textwidth]{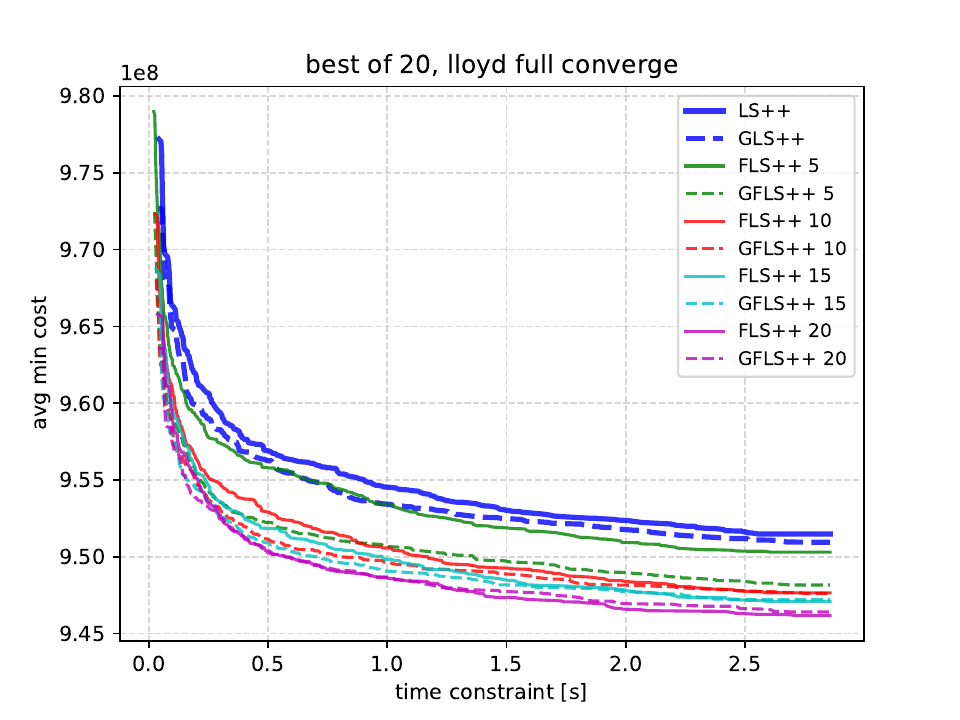}
    \caption{Dataset \textit{pr91}, $k=50$. }
    \end{subfigure}
    \caption{Fixing number of local seach steps for LS++ and GLS++ while using a variable number for FLS++ and GFLS++. The number after either FLS++ or GFLS++ represents the value $Z$ for this algorithm.}
    \label{fig:same_ls_steps}
\end{figure}

\begin{figure}[t]
  \centering
        \includegraphics[width=0.9\textwidth]{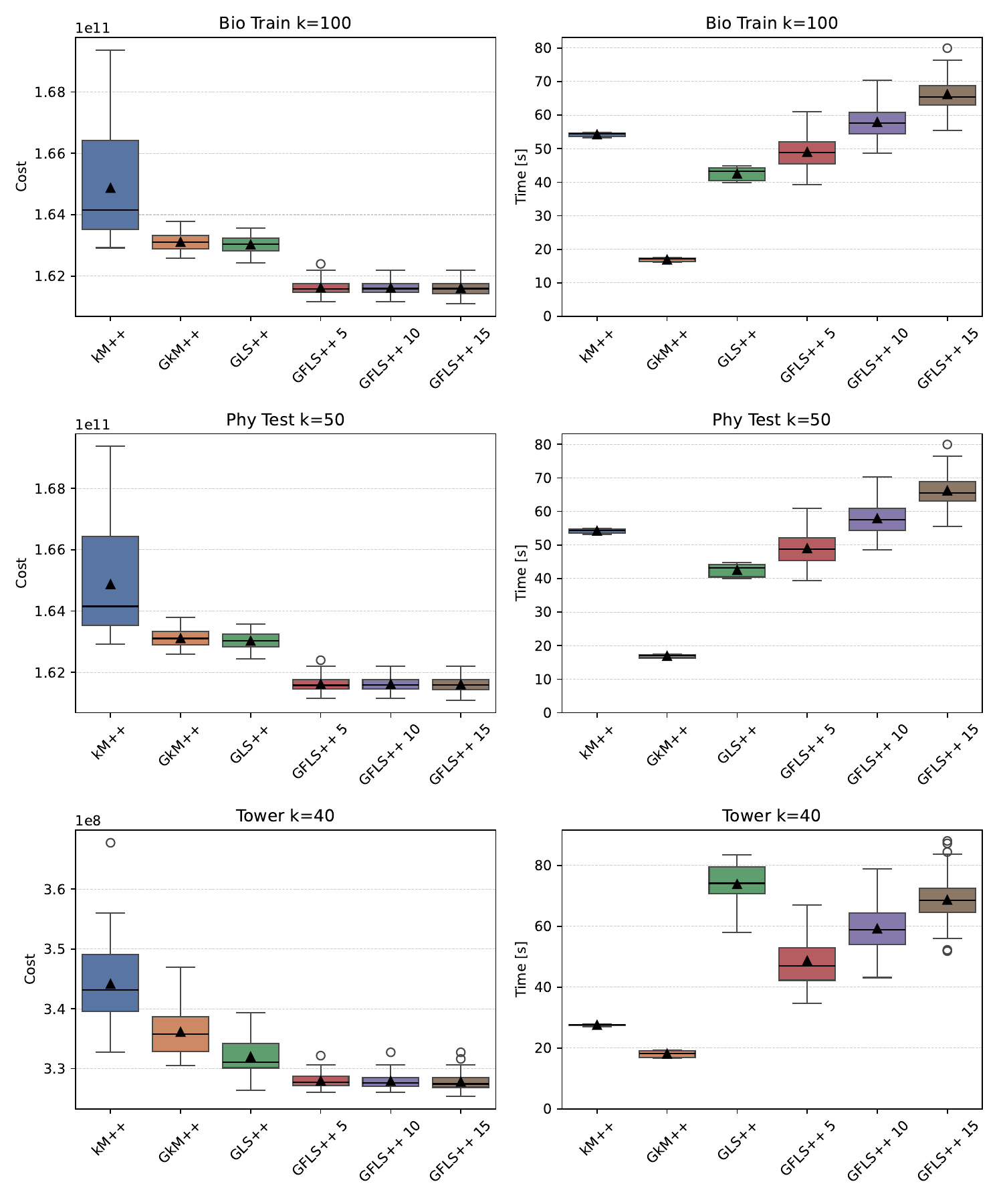}
    \caption{Performance boxplots for different values of $k$ and one additional dataset.}
\label{fig:boxplot_appendix}
\end{figure}